\documentclass[10pt, conference, compsocconf]{IEEEtran}

\ifCLASSINFOpdf

\else

\fi

\usepackage{graphicx}
\usepackage{capt-of}
\usepackage{amssymb}

\usepackage[space]{cite}

\usepackage{url}

\usepackage[cmex10]{amsmath}

\usepackage{algorithmic}

\usepackage{array}
\usepackage{algorithmic}
\usepackage[ruled,norelsize]{algorithm2e}

\usepackage{multirow}
\usepackage[english]{babel}
\usepackage{amsmath}
\usepackage{amsthm}
\usepackage{url}
\newtheorem{lemma}{Lemma}
\newtheorem{theorem}{Theorem}

\usepackage[dvipsnames]{xcolor}
\usepackage{latexsym,color}
\usepackage[tight,footnotesize]{subfigure}

\newenvironment{definition}[1][Definition]{\begin{trivlist}
\item[\hskip \labelsep {\bfseries #1}]}{\end{trivlist}}

\usepackage{mdwmath}
\usepackage{mdwtab}
\newcommand\myeq{\mathrel{\stackrel{\makebox[0pt]{\mbox{\normalfont\tiny def}}}{=}}}

\title{}

\author{}

\begin{document}
\title{Chronos: A Unifying Optimization Framework for Speculative Execution of Deadline-critical MapReduce Jobs}

\author{Maotong Xu, Sultan Alamro, Tian Lan and
	Suresh Subramaniam\\
	Department of ECE, the George Washington University\\
	\{htfy8927, alamro, tlan, suresh\}@gwu.edu}

\maketitle

%
\IEEEpeerreviewmaketitle

\begin{abstract}
Meeting desired application deadlines in cloud processing systems such as MapReduce is crucial as the nature of cloud applications is becoming increasingly mission-critical and deadline-sensitive. It has been shown that the execution times of MapReduce jobs are often adversely impacted by a few slow tasks, known as stragglers, which result in high latency and deadline violations. While a number of strategies have been developed in existing work to mitigate stragglers by launching speculative or clone task attempts, none of them provides a quantitative framework that optimizes the speculative execution for offering guaranteed Service Level Agreements (SLAs) to meet application deadlines. In this paper, we bring several speculative scheduling strategies together under a unifying optimization framework, called {\em Chronos}, which defines a new metric, Probability of Completion before Deadlines (PoCD), to measure the probability that MapReduce jobs meet their desired deadlines. We systematically analyze PoCD for popular strategies including Clone, Speculative-Restart, and Speculative-Resume, and quantify their PoCD in closed-form. The result illuminates an important tradeoff between PoCD and the cost of speculative execution, measured by the total (virtual) machine time required under different strategies. We propose an optimization problem to jointly optimize PoCD and execution cost in different strategies, and develop an algorithmic solution that is guaranteed to be optimal. Chronos is prototyped on Hadoop MapReduce and evaluated against three baseline strategies using both experiments and trace-driven simulations, achieving $50\%$ net utility increase with up to $80\%$ PoCD and $88\%$ cost improvements.
\end{abstract} 

\section{INTRODUCTION}
Distributed cloud computing frameworks, such as MapReduce, have been widely employed by social networks, financial operations and big data analytics due to their ability to process massive amounts of data by splitting large jobs into smaller, parallel tasks. Such frameworks are known to be susceptible to heavy tails in response time. The challenge arises from the fact that execution times of MapReduce jobs are often adversely impacted by a few slow tasks, causing the jobs to possibly suffer high latency and miss their application deadlines.

Prior work has reported that these slow tasks could run up to 8 times slower than the median task~\cite{mapreduce2004, LATE2008, Mantri2010, clone2013}. These slow tasks, known as {\em stragglers}, could significantly impact the overall performance of deadline-sensitive cloud applications and result in the violation of Service Level Agreements (SLAs).

Stragglers are inevitable in cloud environments due to a number of reasons. First, data center computing nodes are typically composed of commodity components, which could be heterogeneous in nature and thus cause various nodes to perform differently. Second, various sources of hardware/software errors exist in large-scale data centers and could lead to node failures interrupting task execution on the failed nodes. Finally, with virtualization and resource sharing, co-scheduled tasks running simultaneously on the same physical machines could create an environment with resource contention and network congestion, which are shown to contribute significantly towards the occurrence of stragglers \cite{LATE2008, Mantri2010, clone2013}. It has been observed that links in data centers could remain under congestion for up to several hundreds of seconds~\cite{kandula2009nature}.

Hadoop has a speculation mode available (which we call Hadoop-S) to mitigate the straggler effect. In this mode, Hadoop can launch extra attempts for map (reduce) tasks after at least one map (reduce) task from the same job has finished. Periodically, Hadoop calculates the difference between the estimated completion time of each running task and the average completion time of finished tasks. It launches one extra attempt for the task that has the largest difference. Thus, if task sizes have large variation, Hadoop-S will launch extra attempts for a large number of tasks, and waste cluster resources. On the other hand, if task sizes are uniform, extra attempts are rarely launched; and so stragglers continue to straggle.

Several better approaches, both proactive and reactive, have been proposed recently to deal with the straggler problem~\cite{Xu2016,parallel2015,yadwadkar2014wrangler,LATE2008,Mantri2010,clone2013,ibrahim2012maestro,rosen2012fine}. In particular, reactive approaches typically detect stragglers by monitoring the progress of tasks and then launch speculative copies of slow-running tasks. For example, Dryad~\cite{Dryad2007} employs a heuristic scheme to detect tasks that are running slower than others, and schedules duplicate attempts. LATE~\cite{LATE2008} proposes a scheduling algorithm to launch speculative copies based on progress score. The progress score is the fraction of data processed. Later, Mantri~\cite{Mantri2010} presents techniques to detect stragglers and act based on their causes. Dolly~\cite{clone2013} is a proactive cloning approach. It launches replica task clones along with the originals and before straggler occurs, to avoid waiting and speculation altogether. Wrangler~\cite{yadwadkar2014wrangler} applies a statistical learning technique and predicts stragglers before they occur and aborts their launch. 
Of the existing approaches, Mantri achieves a large amount of reduction in job completion time and resource usage compared with LATE~\cite{LATE2008} and Dryad~\cite{Dryad2007}. In Mantri, if there is an available container and there is no task waiting for a container, it keeps launching new attempts for a task whose remaining execution time is 30 sec larger than the average task execution time, until the number of extra attempts of the task is larger than 3. Mantri also periodically checks the progress of each task, and leaves one attempt with the best progress running. However, all these existing works only focus on mitigating stragglers without considering deadlines, and therefore cannot provide any guarantee to meet individual application deadlines which can vary significantly in practice. 

Meeting desired deadlines is crucial as the nature of cloud applications is becoming increasingly mission-critical and deadline-sensitive~\cite{cheng2015resource,li2015dcloud}. In this paper, we bring various scheduling strategies together under a unifying optimization framework, called {\em Chronos}, which is able to provide probabilistic guarantees for deadline-critical MapReduce jobs. In particular, we define a new metric, Probability of Completion before Deadlines (PoCD), to measure the probability that MapReduce jobs meet their desired deadlines. Assuming that a single task execution time follows a Pareto distribution \cite{Xu2016, parallel2015, ren2015hopper, xu2017laser}, we analyze three popular classes of (proactive and reactive) strategies, including Clone, Speculative-Restart, and Speculative-Resume, and quantify their PoCD in closed-form under the same framework. 
We note that our analysis also applies to MapReduce jobs, whose PoCD for map and reduce stages can be optimized separately. The result allows us to analytically compare the achievable PoCD of various existing strategies with different context and system configurations. In particular, we show that for the same number of speculative/clone attempts, Clone and Speculative-Resume always outperform Speculative-Restart, and we derive the sufficient conditions under which these strategies attain their highest PoCD. The result also illuminates an important tradeoff between PoCD and the cost of speculative execution, which is measured by the total (virtual) machine time that is required to cover the speculative execution in these strategies. The optimal tradeoff frontier that is characterized in this paper can be employed to determine user's budget for desired PoCD performance, and vice versa. In particular, for a given target PoCD (e.g., as specified in the SLAs), users can select the corresponding scheduling strategy and optimize its parameters, to optimize the required budget/cost to achieve the PoCD target.

The main contributions of this paper are summarized as follows: 
\begin{itemize}
	\item We propose a novel framework, Chronos, which unifies Clone, Speculative-Restart, and Speculative-Resume strategies and enable the optimization of speculative execution of deadline-sensitive MapReduce jobs.
	\item We define PoCD to quantify the probability that a MapReduce job meets its individual application deadline, and analyze PoCD and execution (VM time) cost for the three strategies in closed-form. 
	\item A joint optimization framework is developed to balance PoCD and cost through a maximization of net utility. We develop an efficient algorithm that solves the non-convex optimization and is guaranteed to find an optimal solution. 
	\item Chronos is prototyped on Hadoop MapReduce and evaluated against three baseline strategies using both experiments and trace-driven simulations. Chronos on average outperforms existing solutions by $50\%$ in net utility increase, with up to $80\%$ PoCD and $88\%$ cost improvements.
\end{itemize}

\vspace{-0.15in}
\section{Related Work}


There have been several research efforts to improve the execution time of MapReduce-like systems to guarantee meeting Quality of Service (QoS)~\cite{xu2016theoretical,verma2011aria,lama2012aroma,chen2014cresp,malekimajd2014optimal,herodotou2011no,hwang2012minimizing,kc2010scheduling,polo2010performance,shi2013clotho,zhang2014mimp,abdelbaky2012accelerating,mattess2013scaling,cardosa2011steamengine,rao2013scheduling,polo2011resource,li2014mapreduce,li2014woha,tang2013mapreduce,wang2013slo,chenmapreduce,cheng2015resource,liu2015dreams,palanisamy2015cost,lim2016mrcp,gu2014cost}. Some focus on static resource provisioning to meet a given deadline in MapReduce, while others propose resource scaling in response to resource demand and cluster utilization in order to minimize the overall cost. Moreover,\cite{ahmad2012tarazu,heintz2014cross,ying2015optimizing} proposed frameworks to improve MapReduce job performance. These papers are similar to our proposed work in the sense that resources are optimized to minimize energy consumption and reduce operating cost. However, the above mentioned works do not optimize job execution times in the presence of stragglers, as we do in this work.

Efficient task scheduling is critical to reduce the execution time of MapReduce jobs. A large body of research exists on task scheduling in MapReduce with deadlines. These works range from deadline-aware scheduling~\cite{alamro2016cloud,della2016probabilistic,teng2015mus,tang2015self,li2015packing,bok2016efficient,cheng2015resource,li2015dcloud} to energy- and network-aware scheduling~\cite{li2015sla,wang2015energy,gregory2016constraint,mashayekhy2015energy,wang2016task}. However, these works do not consider dealing with stragglers which might severely prolong a job's execution time and violate the QoS~\cite{yadwadkar2014wrangler,LATE2008,Mantri2010,clone2013}.


The complexity of cloud computing continues to grow as it is being increasingly employed for a wide range of domains such as e-commerce and scientific research. Thus, many researchers have shown interest in mitigating stragglers and improving the default Hadoop speculation mechanism. They proposed new mechanisms to detect stragglers reactively and proactively and launch speculative tasks accordingly~\cite{Xu2016,parallel2015,yadwadkar2014wrangler,LATE2008,Mantri2010,clone2013,ibrahim2012maestro,rosen2012fine}. This is important to ensure providing high reliability to satisfy a given QoS, as it can be at risk when stragglers exist or when failures occur. Different from these works, we jointly maximize the probability of meeting job deadlines and minimize the cost resulting from speculative/duplicate task execution and find the optimal number of speculative copies for each task. Note that in contrast to launching one copy, we launch $r$ extra attempts. Moreover, our work includes introducing an accurate way to estimate the task finishing time by taking the JVM launching time into account, which in turn reduces the number of false positive decisions in straggler detection.

In addition, a number of works have considered avoiding re-executing the work done by the original tasks (stragglers)~\cite{wang2015improving,wang2016betl,quiane2011rafting,wang2015cracking}. Basically, the key idea is to checkpoint running tasks and make the speculative task start from the checkpoint. This idea is similar to our Speculative-Resume strategy. However, in our strategy, we detect stragglers based on jobs' deadlines and launch the optimal number of extra attempts for each straggler in order to jointly optimize PoCD and cost. To save cost, we check the progress of all attempts and keep the fastest one when attempts' progress can be relatively accurately estimated. To save overhead, we only check progress of tasks two times, once for detecting stragglers, and another time to kill slower extra attempts.

\vspace{-0.1in}
\section{Background and System Model}

Consider $M$ MapReduce jobs that are submitted to a datacenter, where job $i$ is associated with a deadline $D_i$ and consists of $N_i$ tasks for $i=1,2,\ldots,M$. Job $i$ meets the desired deadline if all its $N_i$ tasks are processed by the datacenter before time $D_i$.\footnote{When speaking in the context of a single task or job, we drop the subscript(s) for clarity.} Tasks whose completion time exceed $D$ are considered as stragglers. To mitigate stragglers, we launch multiple parallel attempts for each task belonging to job $i$, including one original attempt and $r_i$ speculative/extra attempts. A task is completed as soon as one of its $r_i+1$ attempts is successfully executed. Let $T_{i,j,k}$ denote the (random) execution time of attempt $k$ of job $i$'s task $j$. Thus, we find job completion time $T_i$ and task completion time $T_{i,j}$ by:
\begin{eqnarray}
T_i = \max_{j=1,\ldots,N_i} T_{i,j}, \ {\rm where} \ T_{i,j}=\min_{k=1,\ldots,r_i+1} T_{i,j,k}, \ \forall j.
\end{eqnarray}

The Pareto distribution is proposed to model the execution times of tasks in~\cite{grass2014}, and is used in~\cite{Xu2016, parallel2015, ren2015hopper} to analyze the straggler problem. Following these papers, we assume the execution time $T_{i,j,k}$ of each attempt follows a Pareto distribution with parameters $t_{\min}$ and $\beta$. The probability density function of $T_{i,j,k}$ is
\vspace{-0.1in}
\begin{equation}
 f_{T_{i,j,k}}(t)=
 \begin{cases}
 \frac{\beta{\cdot}t^{\beta}_{min}}{t^{\beta+1}} & t{\geq}t_{min}, \\
 0 & t<t_{min},
 \end{cases}
\end{equation}
where $t_{\min}$ is the minimum execution time and $\beta$ is the tail index, while different attempts are assumed to have independent execution times. 

As introduced earlier, Chronos consists of three scheduling strategies to mitigate stragglers, namely Clone, Speculative-Restart, and Speculative-Resume. Clone is a proactive approach, wherein $r + 1$ copies of a task (i.e., one original attempt and $r$ extra attempts) are launched simultaneously. Here, $r$ is a variable that is optimized to balance the PoCD with the cost of execution. Speculative-Restart launches $r$ extra copies for each straggler detected. 
%
Speculative-Resume is a work-preserving strategy, where the $r$ extra copies of a detected straggler start processing right after the last byte offset processed by the original task. To model the cost of executing job $i$ under each strategy, we consider an on-spot price of $\gamma_i$ dollars per unit time for each active virtual machine (VM) running attempts/tasks of job $i$. The price $\gamma_i$ depends on the VM-type subscribed by job $i$, and is assumed to be known when optimizing job scheduling strategies. Our goal is to jointly maximize the probability of meeting job deadlines and minimize the cost resulting from speculative/extra task/attempt scheduling. In Chronos, we use progress score to determine if extra attempts are needed. The progress score is defined as the percentage of workload processed at a given time $t$.

\begin{figure*}[!t]
	\subfigure[]{%
		\includegraphics[height=1.6in,width=0.33\textwidth]{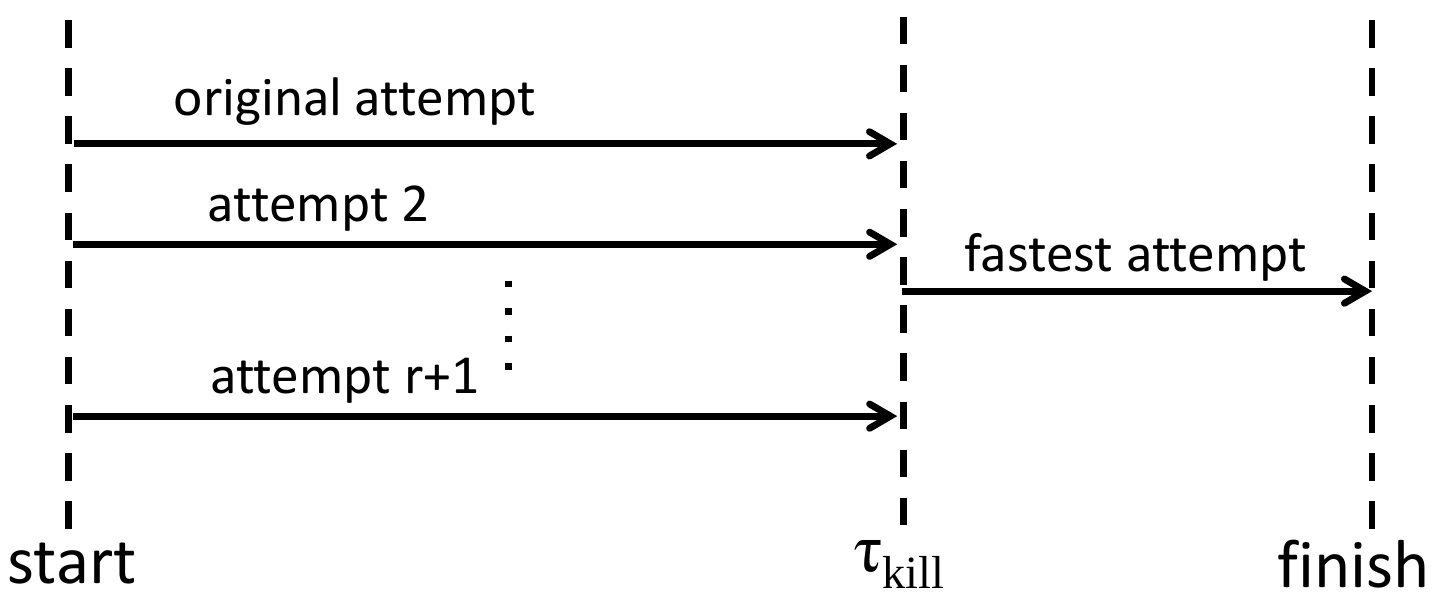}%
		\label{fig:clone}%
	}%
	~
	\centering
	\subfigure[]{%
		\includegraphics[height=1.6in, width=0.33\textwidth]{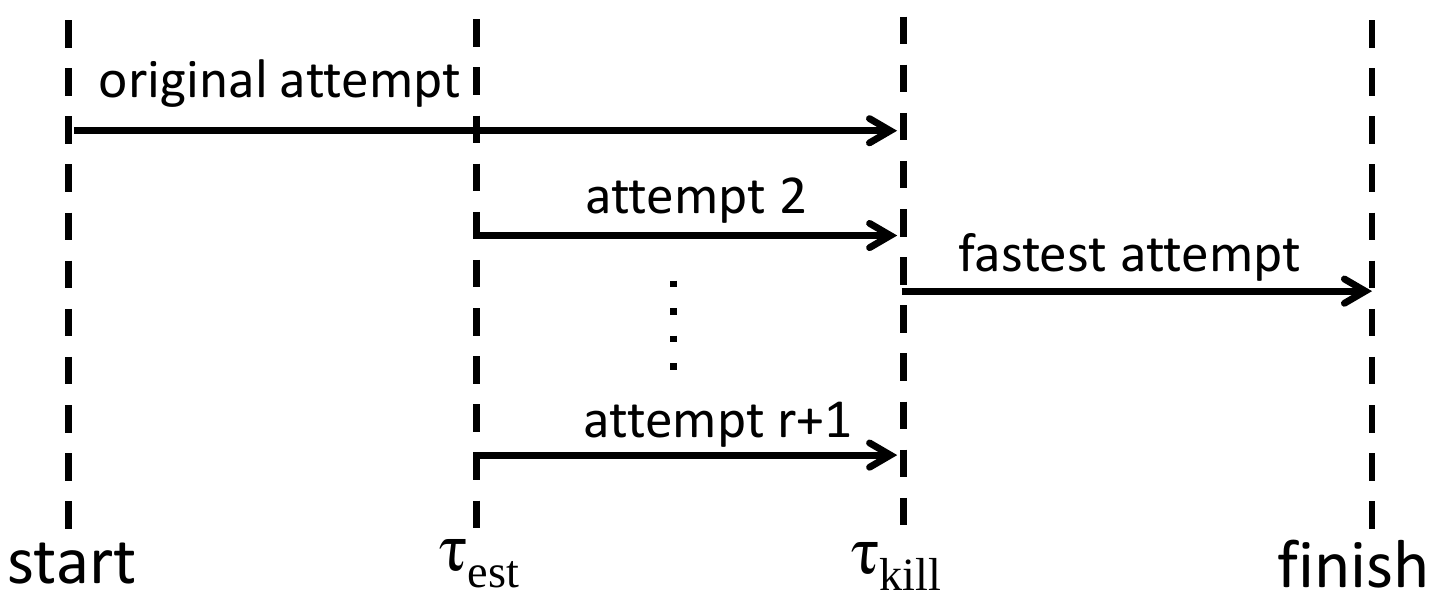}%
		\label{fig:s-restart}%
	}%
	~
	\centering
	\subfigure[]{%
		\includegraphics[height=1.7in,width=0.33\textwidth]{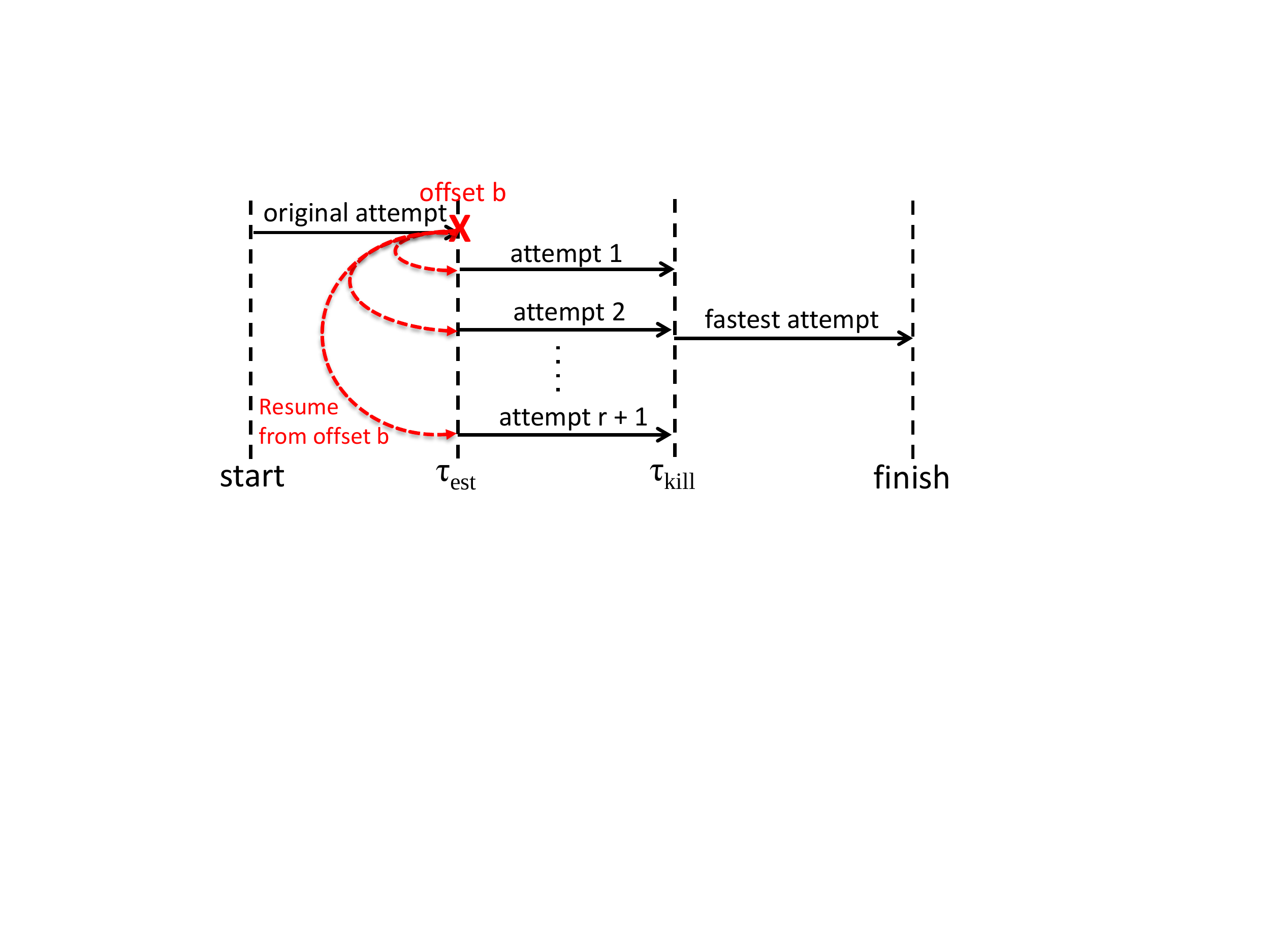}%
		\label{fig:s-resume}%
	}%
	\vspace{-0.12in}
	\caption{(a) Clone Strategy, (b) Speculative-Restart Strategy, (c) Speculative-Resume Strategy.}
	\vspace{-0.17in}
\end{figure*}

\noindent {\bf Clone Strategy.} Under this strategy, $r+1$ attempts of each task are launched at the beginning. The progress scores of the $r+1$ attempts are checked at time $\tau_{\rm kill}$, and the attempt with the best progress is left for processing data, while the other $r$ attempts are killed to save machine running time. Figure \ref{fig:clone} illustrates the Clone strategy for a single task. 

\noindent {\bf Speculative-Restart Strategy.} Under this strategy, one attempt (original) of each task is launched at the beginning. At time $\tau_{\rm est}$, the task attempt completion time is estimated, and if it exceeds $D$, $r$ extra attempts are launched that start processing data from the beginning.\footnote{No extra attempts are launched if the estimated finishing time is not greater than $D$.} At time $\tau_{\rm kill}$, the progress scores of all $r+1$ attempts are checked, and the attempt with the smallest estimated completion time is left running, while the other $r$ attempts are killed. Figure \ref{fig:s-restart} illustrates the Speculative-Restart strategy for a task when the execution time of the original attempt exceeds $D$.

\noindent {\bf Speculative-Resume Strategy.}  This strategy is similar to the Speculative-Restart strategy in its straggler detection. The difference is that the detected straggler is killed and $r+1$ attempts are launched for the straggling task. These attempts, however, do not reprocess the data that has already been processed by the original attempt, and  start processing the data after the last byte offset when the straggler is detected. At time $\tau_{\rm kill}$, the progress scores of all attempts are checked, and the attempt with the smallest estimated completion time is left running while the other $r$ attempts are killed. Figure \ref{fig:s-resume} illustrates the Speculative-Resume strategy for a task for the case when the execution time of the original attempt exceeds $D$. Here, the processed byte offset of the original attempt at $\tau_{\rm est}$ is $b$. Extra attempts launched at $\tau_{\rm est}$ start to process data from byte offset $b$. 

\vspace{-0.03in}
\section{Analysis of PoCD and machine running time}

We now formally define PoCD, and analyze PoCD for each strategy. PoCD expressions for the three strategies, under the assumption that task attempt execution times are iid Pareto, are presented in Theorems \ref{Clone_PoCD}, \ref{restart_PoCD}, and \ref{resume_PoCD}. We also analyze the machine running time for each strategy, and present expressions for these in Theorems \ref{Clone_time}, \ref{restart_PoCD}, and \ref{resume_PoCD}. 

\begin{definition}\label{def_PoCD}
PoCD is the probability that an arriving job completes before its deadline.
\end{definition}
We use the notation $R_{\rm strategy}$ to denote the PoCD of a particular strategy, which can be Clone, Speculative-Restart, or Speculative-Resume.

\subsection{Clone}
We start by analyzing Clone, and derive PoCD and machine running time expressions in {\em Theorem} \ref{Clone_PoCD} and {\em Theorem} \ref{Clone_time}, respectively.

\begin{theorem}\label{Clone_PoCD}
	Under Clone strategy, the PoCD
	\begin{equation}\label{clone1}
	R_{\rm Clone} = \left[1-\left(\frac{t_{\min}}{D}\right)^{\beta{\cdot}(r+1)}\right]^{N}.
	\end{equation}
\end{theorem}
\vspace{-0.1in}
\begin{proof}
We first derive the probability that a single task completes before the deadline, and then derive the PoCD by considering all $N$ tasks of the job.
	
Let us denote the probability that a single attempt of a task fails to finish before deadline $D$ by $P_{\rm Clone}$. Then, we have
	\begin{equation}\label{clone2}
	P_{\rm Clone}=
	\int_{D}^{\infty}{\frac{{\beta}{t^{\beta}_{\min}}}{t^{\beta+1}}}{dt}=\left(\frac{t_{\min}}{D}\right)^{\beta}.
	\end{equation}
	
The task fails to finish before $D$ when all $r+1$ attempts fail to finish before $D$. Thus, the probability that a task finishes before $D$is $1-(P_{\rm Clone})^{r+1}$. The job finishes before deadline $D$ when all its $N$ tasks finish before $D$. Thus, the PoCD is given by
	\vspace{-0.1in}
	\begin{equation}\label{clone3}
	R_{\rm Clone} = [1-(P_{\rm Clone})^{r+1}]^{N} = \left[1-\left(\frac{t_{\min}}{D}\right)^{\beta{\cdot}(r+1)}\right]^{N}.
	\end{equation}
\end{proof}

\vspace{-0.1in}
Before moving to {\em Theorem} \ref{Clone_time}, we first introduce and prove {\em Lemma} \ref{lemma1} for computing the expected execution time of $\min(T_{j,1},...,T_{j,n})$, where $n$ is a positive integer.

\begin{lemma}\label{lemma1}
	Let $W = \min(T_{j,1},...,T_{j,n})$, where $T_{j,a}, a = 1, 2, \ldots, n$ follows the Pareto distribution with parameters $t_{\min}$ and $\beta$. Then, $E(W)$ can be computed as:
	\begin{equation}\label{lemma1_1}
	E(W)=\frac{t_{\min}{\cdot}n{\cdot}\beta}{n{\cdot}\beta-1}.
	\end{equation}
\end{lemma}
\begin{proof}
	We introduce
	\begin{equation}\label{lemma1_2}
	W-t_{\min}=\int_{t_{\min}}^{W}{dt}=\int_{t_{\min}}^{\infty}I(t)dt,
	\end{equation}
	where $I(t)$ is the indicator function for event $\{W>t\}$:
	\begin{eqnarray}\label{lemma1_3}
	I(t)=I\{W>t\}{\myeq}\left\{
	\begin{array}{lll}
	1, \ \ W>t;\\
	0, \ \ W{\leq}t.
	\end{array}
	\right.
	\end{eqnarray}
	Taking expectations we have that
	\begin{equation}\label{lemma1_4}
	E(W)=\int_{t_{\min}}^{\infty}E(I(t))dt+{t_{\min}}=\int_{t_{\min}}^{\infty}P(W>t)dt+{t_{\min}}.
	\end{equation}
	
	Since $P(W>t)=\left[P(T_{j,a}>t)\right]^n$, we have:
	\begin{equation}\label{lemma1_5}
	E(W)=\int_{t_{\min}}^{\infty}\left(\frac{t_{\min}}{t}\right)^{n{\cdot}\beta}dt+{t_{\min}}
	=\frac{t_{\min}{\cdot}n{\cdot}\beta}{n{\cdot}\beta-1}
	\end{equation}
\end{proof}
\vspace{-0.1in}
\begin{theorem}\label{Clone_time}
	Under Clone strategy, the expected execution time of a job
	\begin{equation}\label{clone4}
	E_{\rm Clone}(T) = N{\cdot}\left[r{\cdot}{\tau_{\rm kill}}
	+t_{\min}+\frac{t_{\min}}{{\beta}{\cdot}(r+1)-1}\right],
	\end{equation}
	where $T$ denotes execution time of a job.
\end{theorem}

\begin{proof}
We first derive the machine running time of a task by adding the machine running time of attempts killed at $\tau_{\rm kill}$ and the machine running time of the attempt that successfully completes, and then get the job's machine running time by adding the machine running time of all $N$ tasks.
	
	$E_{\rm Clone}(T)$ equals the expectation of the machine running time of all $N$ tasks, i.e., $E_{\rm Clone}(T) = N{\cdot}E(T_j)$, where $T_j$ is the machine running time of task $j$. Also, $E(T_j)$ equals the sum of the machine running times of the $r$killed attempts and the execution time of the attempt with the best progress score at $\tau_{\rm kill}$. Denote ${\min}\{T_{j,1},...,T_{j,r+1}\}$ as $W^{\rm all}_{j}$, where $T_{j,a}$ is the execution time of attempt $a$ belonging to task $j$. Then,
	\begin{equation}\label{clone5}
	E(T_j)=r{\cdot}{\tau_{\rm kill}}+E(W^{\rm all}_{j}).
	\end{equation}
	
	We can use the result from {\em Lemma} \ref{lemma1} to compute $E(W^{\rm all}_{j})$. The parameter $n$ in {\em Lemma} \ref{lemma1} equals $r+1$. Thus,
	\begin{equation}\label{clone6}
	E(W^{\rm all}_{j})=\frac{t_{\min}{\cdot}{\beta}{\cdot}(r+1)}{{\beta}{\cdot}(r+1)-1}=t_{\min}+\frac{t_{\min}}{{\beta}{\cdot}(r+1)-1}.
	\end{equation}
\end{proof}

\vspace{-0.15in}
\subsection{Speculative-Restart}

Here we present the PoCD and machine running time analysis for the Speculative-Restart strategy in {\em Theorem} \ref{restart_PoCD} and {\em Theorem} \ref{restart_time}, respectively.

\begin{theorem}\label{restart_PoCD}
	Under Speculative-Restart, the PoCD
	\begin{equation}\label{restart0}
	R_{\rm S-Restart}=\left[1-\frac{t^{\beta{\cdot}(r+1)}_{\min}}{D^{\beta}{\cdot}(D-\tau_{\rm est})^{\beta{\cdot}r}}\right]^N.
	\end{equation}
\end{theorem}

\begin{proof}
	As in the proof of Theorem~\ref{Clone_PoCD}, we first derive the probability that a task completes before the deadline. Then, we can obtain the PoCD by considering that all $N$ tasks of the job complete before the deadline.
\end{proof}

\vspace{-0.1in}
\begin{theorem}\label{restart_time}
	Under Speculative-Restart, the expected execution time of a job, $E_{\rm S-Restart}(T)$, equals
	\begin{equation}\label{restart4}
	E(T_j|T_{j,1}{\leq}D){\cdot}P(T_{j,1}{\leq}D)+E(T_j|T_{j,1}{>}D){\cdot}P(T_{j,1}{>}D),
	\end{equation}
	where
	\vspace{-0.1in}
	\begin{multline}
	\qquad \quad P(T_{j,1}{>}D) = 1-P(T_{j,1}{\leq}D)=\left(\frac{t_{\min}}{D}\right)^{\beta},\\
	E(T_j|T_{j,1}{\leq}D)  =\frac{t_{\min}{\cdot}D{\cdot}\beta{\cdot}(t^{\beta-1}_{\min}-D^{\beta-1})}{(1-\beta){\cdot}(D^{\beta}-t^{\beta}_{\min})},\\
	E(T_j|T_{j,1}{>}D) = \tau_{\rm est}+r{\cdot}(\tau_{\rm kill}-\tau_{\rm est}) \\ +\frac{t_{\min}}{\beta{\cdot}r-1}-\frac{t^{\beta{\cdot}r}_{\min}}{(\beta{\cdot}r-1){\cdot}(D-\tau_{\rm est})^{\beta{\cdot}r-1}} \\
	+\int_{D-\tau_{\rm est}}^{\infty}\left(\frac{D}{\omega+\tau_{\rm est}}\right)^{\beta}{\cdot}\left(\frac{t_{\min}}{\omega}\right)^{\beta{\cdot}r}d\omega
	+t_{\min} \label{eq:cases}
	\end{multline} 
\end{theorem}
\begin{proof}
We first derive the machine running time of a task by considering if the execution time of the original attempt is larger than $D$. If the execution time is no more than $D$ (denoted by $E(T_j|T_{j,1}{\leq}D)$ in (\ref{eq:cases})), there is no extra attempt launched, and the machine running time is the execution time of the original attempt. 

In the case that execution time is larger than D, the machine running time (denoted by $E(T_j|T_{j,1}>D)$ in (\ref{eq:cases})) consists of three parts, i.e., (i) execution of the original attempt from start to $\tau_{est}$, (ii) machine time need to run r+1 attempts between $\tau_{est}$ and  $\tau_{kill}$, and (iii) execution of the fastest attempt from $\tau_{kill}$ until it finishes. Due to space limitation, we omit the proof details and refer to Section~\ref{appendix}.
\end{proof}

\subsection{Speculative-Resume}

Finally, we obtain PoCD and machine running time expressions for Speculative-Resume in {\em Theorems} \ref{resume_PoCD} and \ref{resume_time}, respectively. Let us denote the average progress of the original attempts at time $\tau_{\rm est}$ as $\varphi_{j, \rm est}$.

\begin{theorem}\label{resume_PoCD}
	Under Speculative-Resume, the PoCD
	\begin{equation}\label{resume1}
	R_{\rm S-Resume} = \left[1-\frac{(1-\varphi_{j, \rm est})^{{\beta}{\cdot}(r+1)}{\cdot}t^{{\beta}\cdot(r+2)}_{\min}}{D^{\beta}{\cdot}(D-\tau_{\rm est})^{{\beta}{\cdot}(r+1)}}\right]^N.
	\end{equation}
\end{theorem}

\begin{proof}
The proof is similar to the proof of Theorem~\ref{Clone_PoCD}, and details are omitted. 
\end{proof}


\begin{theorem}\label{resume_time}
	Under Speculative-Resume, the expected execution time of a job $E_{\rm S-Resume}(T)$ equals
	\begin{equation}
	E(T_j|T_{j,1}{\leq}D){\cdot}P(T_{j,1}{\leq}D)+E(T_j|T_{j,1}{>}D){\cdot}P(T_{j,1}{>}D),
	\end{equation}
	where
	\begin{align}
	P(T_{j,1}{>}D) & = 1-P(T_{j,1}{\leq}D)=\left(\frac{t_{\min}}{D}\right)^{\beta}, \\
	E(T_j|T_{j,1}{\leq}D) & =\frac{t_{\min}{\cdot}D{\cdot}\beta{\cdot}(t^{\beta-1}_{\min}-D^{\beta-1})}{(1-\beta){\cdot}(D^{\beta}-t^{\beta}_{\min})}, {\rm and} \\
	E(T_j|T_{j,1}{>}D) & = \tau_{\rm est}+r{\cdot}(\tau_{\rm kill}-\tau_{\rm est}) \\
	\ & +\frac{t_{\min}{\cdot}(1-\varphi_{j,\rm est})^{\beta{\cdot}(r+1)}}{{\beta}{\cdot}(r+1)-1}+t_{\min}
	\end{align}
\end{theorem}
\begin{proof}
The proof is similar to the proof of Theorem~\ref{restart_time}, and is omitted.
\end{proof}

We note that our analysis of PoCD and cost (including proof techniques of Theorems 1-6) actually works with other distributions as well (even though the exact equations/values may change in the derived results, depending on the distribution).

\subsection{Comparing PoCD of different strategies}\label{subcompare}
In this subsection, we compare the PoCDs of three strategies, and present results in Theorem \ref{compare}. We denote $D-\tau_{\rm est}$ as $\overline{D}$, and $1-\varphi_{j,\rm est}$ as $\overline{\varphi}_{j,\rm est}$.

\begin{theorem}\label{compare}
	Given $r$, we can get three conclusions:
	\begin{enumerate}
		\item $R_{\rm clone}>R_{\rm S-Restart}$,
		\item $R_{\rm S-Resume}>R_{\rm S-Restart}$,
		\item if $r>\frac{{\beta}{\cdot}\ln({\overline{\varphi}_{j, \rm est}}{\cdot}t_{\min})-\ln\overline{D}}{\ln\overline{D}-\ln(\overline{\varphi}_{j, \rm est}{\cdot}D)}$,
		$R_{\rm clone}>R_{\rm S-Resume}$;  \\
		Otherwise, $R_{\rm clone}{\leq}R_{\rm S-Resume}$.
	\end{enumerate}
\end{theorem}

We provide an outline of our proof. Given $r$, Clone strategy launches $r$ extra attempts for each task from the beginning.
For each straggler detected at $\tau_{est}$, S-Restart strategy launches $r$ extra new attempts (that restart from zero) and keeps the original attempt, while S-Resume strategy launches $r$ extra attempts, which resume to process the remaining $\overline{\varphi}_{j,est}$ data, for each straggler. If $r$ is not large, it is clearly better to kill the straggler, and launch extra $r+1$ attempts to process remaining workload, instead of leaving straggler running. 

\vspace{-0.08in}
\section{Joint PoCD and Cost Optimization}\label{sec:joint_opt}
Starting multiple speculative/clone tasks leads to higher PoCD, but also results in higher execution cost, which is proportional to the total VM time required for both original and speculative/clone tasks, e.g., when VMs are subscribed from a public cloud with a usage-based pricing mechanism. To exploit this tradeoff, we consider a joint optimization framework for the three proposed strategies, to maximize a ``net utility" defined as PoCD minus execution cost. More precisely, we would like to:
\vspace{-0.08in}
\begin{eqnarray}\label{net_utility}
& \max & U(r) = f(R(r)-R_{\min})-\theta{\cdot}C{\cdot}\mathbb{E}(T), \\
& {\rm s.t.} &  r\ge 0, \\
& {\rm var.} & r\in\mathbb{Z}.
\end{eqnarray}
Here, $R(r)$ is the PoCD that results by initiating $r$ speculative/clone tasks in Clone, S-Restart, or S-Resume strategies. To guarantee a minimum required PoCD, we consider a utility function $f(R(r)-R_{\min})$, which is an increasing function and drops to negative infinity if $R(r)<R_{\min}$. The usage-based VM price per unit time is $C$.

We use a tradeoff factor $\theta\ge 0$ to balance the PoCD objective $f(R(r)-R_{\min})$ and the execution cost $C{\cdot}\mathbb{E}(T)$. By varying $\theta$, the proposed joint optimization can generate a wide range of solutions for diverse application scenarios, ranging from PoCD-critical optimization with small tradeoff factor $\theta$ and cost-sensitive optimization with large $\theta$. Finally, while our joint optimization framework applies to any concave, increasing utility function $f$, in this paper, we focus on logarithmic utility functions, $f(R(r)-R_{\min})=\lg(R(r)-R_{\min})$, which is known to achieve proportional fairness~\cite{CORA2015}. In the following, we will prove the concavity of the optimization objective $U(r)$ for different strategies and propose an efficient algorithm to find the optimal solution to the proposed joint PoCD and cost optimization.
\vspace{-0.08in}
\subsection{Optimizing Chronos}

To evaluate convexity of the optimization objective $U_{\rm strategy}(r)$, we first present the following lemma on function compositions.
\vspace{-0.08in}
\begin{lemma}\label{lemma3} \cite{boydConvex}
	Suppose $H(x)$ is a composite of $f(x)$ and $g(x)$, i.e., $H(x)=f(g(x))$. If $f(x)$ is increasing and both $f(x)$ and $g(x)$ are concave, then $H(x)$ is a concave function.
\end{lemma}
\vspace{-0.1in}
\begin{theorem}\label{opt_clone} The optimization objective $U_{\rm strategy}(r)$ under Clone, S-Restart, and S-Resume, i.e.,
	\begin{equation}\label{jo2}
	U_{\rm strategy}(r) = \lg(R_{\rm strategy}(r)-R_{\min})-\theta{\cdot}C{\cdot}E_{\rm strategy}(T),
	\end{equation}
	is a concave function of $r$, when $r > \Gamma_{\rm strategy}$, where
	\begin{equation}\label{jo3}
	\Gamma_{\rm Clone} = -\beta^{-1}{\cdot}\log_{t_{\min}/D}{N}-1.
	\end{equation}
	\vspace{-0.1in}
	\begin{equation}\label{jo6}
	\Gamma_{\rm S-Restart} = \beta^{-1}{\cdot}\log_{t_{\min}/(D-\tau_{\rm est})}\frac{D^{\beta}}{N{\cdot}t^{\beta}_{\min}}.
	\end{equation}
	\vspace{-0.1in}
	\begin{equation}\label{jo10}
	\Gamma_{\rm S-Resume}= \beta^{-1}{\cdot}\log_{\frac{(1-\varphi_{\rm j,est}){\cdot}t_{\min}}{D-\tau_{\rm est}}}\frac{D^{\beta}}{N{\cdot}t^{\beta}_{\min}}-1.
	\end{equation}
	\vspace{-0.15in}
\end{theorem}

\begin{proof}
	$\lg(R_{\rm strategy}(r)-R_{\min})$ is an increasing and concave function of $R_{\rm strategy}(r)$. Also, when $r{>}\Gamma_{\rm strategy}$, the second derivative of $R_{\rm strategy}$ is less than $0$. So, $R_{\rm strategy}$ is a concave function of $r$, when $r{>}\Gamma_{\rm strategy}$. Based on {\em Lemma} \ref{lemma3}, we know that $\lg(R_{\rm Clone}-R_{\min})$ is a concave function of $r$, when $r{>}\Gamma_{\rm strategy}$.
	
The second derivative of $-E_{\rm strategy}(T)$ is less than $0$. Since the sum of two concave functions is also a concave function, $U_{\rm strategy}(r)$ is a concave function when $r{>}\Gamma_{\rm strategy}$.
\end{proof}

We note that for non-deadline sensitive jobs, as job deadlines increase and become sufficiently large, the optimal $r$ will approach zero. Thus, there is no need to launch speculative or clone attempts for these non-deadline sensitive jobs.

\subsection{Our Proposed Optimization Algorithm}
We present the algorithm for solving the optimal $r$ to maximize the ``PoCD minus cost" objective for the three strategies. First, our analysis shows that the optimization objective is guaranteed to be concave when $r$ is above a certain threshold, i.e., $r>\Gamma_{\rm strategy}$, for Clone, Speculative-Restart, and Speculative-Resume strategies. Thus, we can leverage the concavity property and employ a gradient descent method to find the optimal objective $U_{\rm strategy,opt}(r_{\rm opt})$ and solution $r_{\rm opt}$, when the objective is concave, i.e., $r>\Gamma_{\rm strategy}$ for different threshold $\Gamma_{\rm strategy}$. Second, for the non-concave regime, $r\le \Gamma_{\rm strategy}$, which contains a limited number of integer points (typically, less than $4$, as we will see later), we use a simple search to find the maximum objective value over all integers $r\le \Gamma_{\rm strategy}$. This hybrid algorithm can efficiently find the optimal solution to the joint PoCD and cost optimization, due to the optimality of convex optimization when $r>\Gamma_{\rm strategy}$, as well as the limited complexity for searching $r\le \Gamma_{\rm strategy}$. The proposed algorithm is summarized in Algorithm \ref{alg0}. In Algorithm \ref{alg0}, $\eta$, $\alpha$ and $\xi$ determine the accuracy and complexity of the algorithm, and their values depend on the values of $N$, $\theta$, $\tau_{\rm est}$, $\tau_{\rm kill}$ and $C$.

\begin{theorem}\label{optimality}
	Algorithm \ref{alg0} finds an optimal solution for the net utility optimization problem in (\ref{net_utility}).
\end{theorem}
\begin{proof}
	For $r{\geq}{\lceil}\Gamma_{\rm strategy}{\rceil}$, $U_{\rm Strategy}(r)$ is a concave function of $r$, and the {Gradient-based Line Search} can find the maximum value of $U_{\rm Strategy}(r)$ in polynomial time, where {\em strategy} can be either Clone, S-Restart, or S-Resume. We can then compute the $U_{\rm Strategy}(r)$, $\forall \ r\in{0,...,{\lceil}\Gamma_{\rm strategy}{\rceil}-1}$, and compare the result from {Gradient-based Line Search} to get the maximum value of $U_{\rm Strategy}(r)$, $\forall \ r{\geq}0$. Since $\Gamma_{\rm strategy}$ is typically a small number, Algorithm \ref{alg0} can find the optimal $r$ for maximizing net utility.
\end{proof}

\begin{algorithm}                      
	\caption{{Unifying Optimization Algorithm}}          
	\label{alg0}                           
	\begin{algorithmic}
		\STATE{strategy=\{Clone, S-Restart, S-Resume\}}
		\STATE{}//{Phase 1 : {Gradient-based Line Search}\cite{boydConvex}}
		\STATE{$r=\lceil\Gamma_{\rm strategy}\rceil$}
		\WHILE {$|{\nabla}U_{\rm Strategy}(r)| {>} \eta$}
		\STATE{${\triangle}r=-{\nabla}U_{\rm Strategy}(r)$}
		\STATE{$\varepsilon=1$}
		\WHILE {$U_{\rm Strategy}(r+\varepsilon{\cdot}{\triangle}r)>U_{\rm Strategy}(r)
			+\alpha{\cdot}\varepsilon{\cdot}{\nabla}U_{\rm Strategy}(r){\cdot}{\triangle}r$}
		\STATE{$\varepsilon = \xi{\cdot}\varepsilon$}
		\ENDWHILE
		\STATE{$x=x+t{\cdot}{\triangle}r$}
		\ENDWHILE
		\STATE{$U_{\rm strategy,opt}(r_{\rm opt})=U_{\rm strategy}(r)$}
		\STATE{$r_{\rm opt}=r$}
		\STATE{}//{Phase 2 : {Finding optimal solution for $r{\geq}0$}}
		\FOR{$r=0:\lceil\Gamma_{\rm strategy}\rceil -1$}
		\STATE{Compute utility function $U_{\rm strategy}(r)$}
		\IF{$U_{\rm strategy}(r)>U_{\rm strategy,opt}(r_{\rm opt})$}
		\STATE{$U_{\rm strategy,opt}(r_{\rm opt})=U_{\rm strategy}(r)$}
		\STATE{$r_{\rm opt}=r$}
		\ENDIF
		\ENDFOR
	\end{algorithmic}
\end{algorithm}

\vspace{-0.13in}
\section{Implementation}
We develop Chronos (including Clone, Speculative-Restart and Speculative-Resume) using Hadoop YARN, which consists of a central Resource Manager (RM), a per-application Application Master (AM) and a per-node Node Manager (NM). The AM negotiates resources from the RM and works with the NMs to execute and monitor an application's tasks. Our optimization algorithm is implemented in the AM to calculate the optimal $r$ upon job submission to be used by all strategies. Chronos uses the progress score which is provided by Hadoop to estimate the expected completion time of a given task. The progress score of the map phase represents the fraction of data processed. In Chronos, we modified the default Hadoop speculation method and implemented our new progress estimation mechanism, which significantly improves the estimation accuracy by taking into account the time overhead for launching MapReduce tasks, which is ignored by Hadoop. 




\subsection{Implementing Clone Strategy}
Upon submission of a job, the AM creates tasks and queues them for scheduling. Then, the AM asks the RM for containers to run its tasks. In Clone, the AM solves the proposed joint optimization problem to determine the optimal $r$ before creating tasks for a given job. As discussed in Section \ref{sec:joint_opt}, $r$ is the number of extra attempts for each task. Once $r$ is found, whenever a map task belonging to the job is created, the AM also creates extra $r$ copies of the same task, which start execution simultaneously with the original one. After launching the $r+1$ attempts, the AM monitors all attempts' progress scores, and after $\tau_{\rm kill}$, the attempt with the highest progress score is kept alive and will continue running, while all other slower clone tasks are killed.

\subsection{Implementing Speculation Strategies}
Unlike Hadoop, Chronos launches $r$ speculative attempts for each straggler task whose estimated completion time falls behind the target deadline. The optimal number of attempts $r$ is calculated by the AM after a job is submitted, by solving the joint PoCD and cost optimization. The AM monitors all active tasks and launches speculative attempts for the tasks that have an estimated completion time ($t_{\rm ect}$) larger than their target deadline $D$. In particular, we have $t_{\rm ect} = t_{\rm lau} + t_{\rm eet}$, where $t_{\rm lau}$ and $t_{\rm eet}$ are the time at which the task is launched and the estimated execution time, respectively.\footnote{All times are relative to the start time of a job, which is assumed to be $0$. Tasks are launched a little after time 0, namely $t_{\rm lau}$.} In our experiments, we observed that Hadoop's default task completion time estimation is highly unreliable. In Hadoop, the task estimated execution time equals the difference between the current time and the time at which the task is launched, divided by the reported progress score.

\begin{figure*}[!t]
	\centering
	\subfigure[]{%
		\includegraphics[height=1.8in,width=0.32\textwidth]{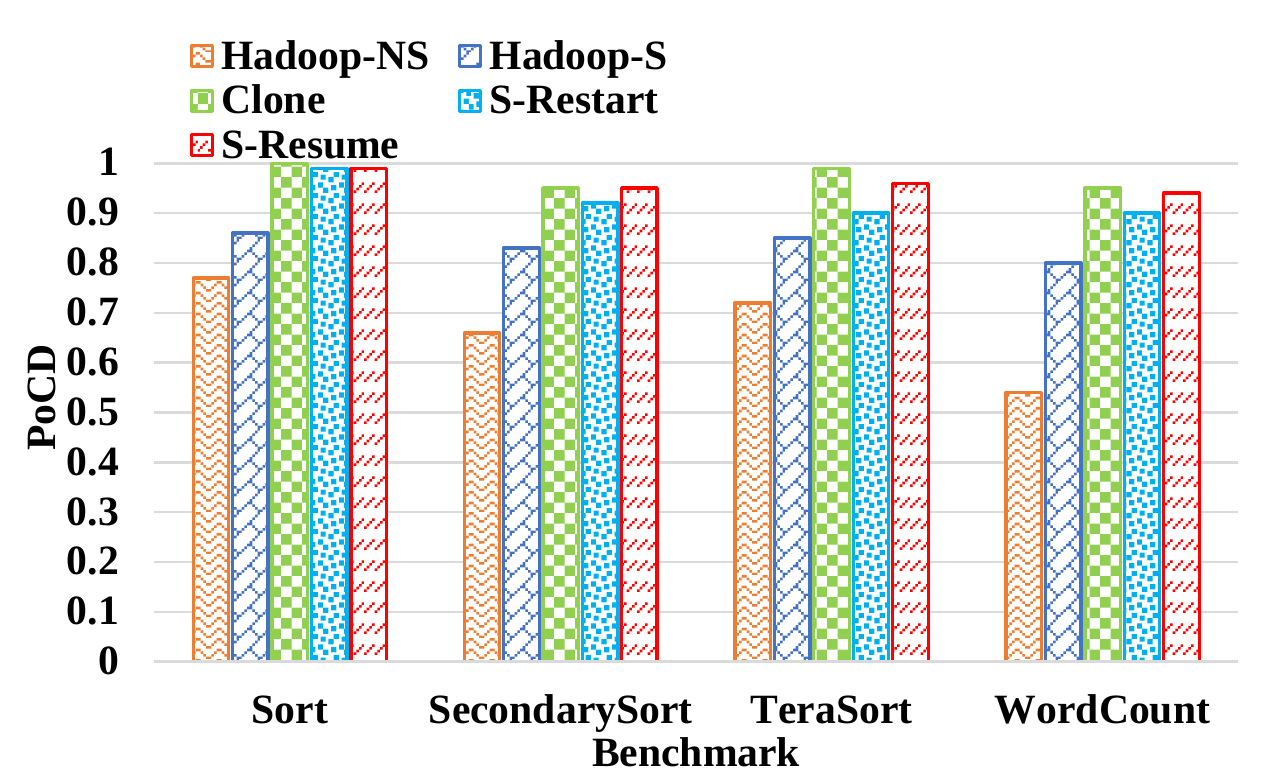}%
		\label{fig:x1}%
	}%
	~
	\subfigure[]{%
		\includegraphics[height=1.8in, width=0.32\textwidth]{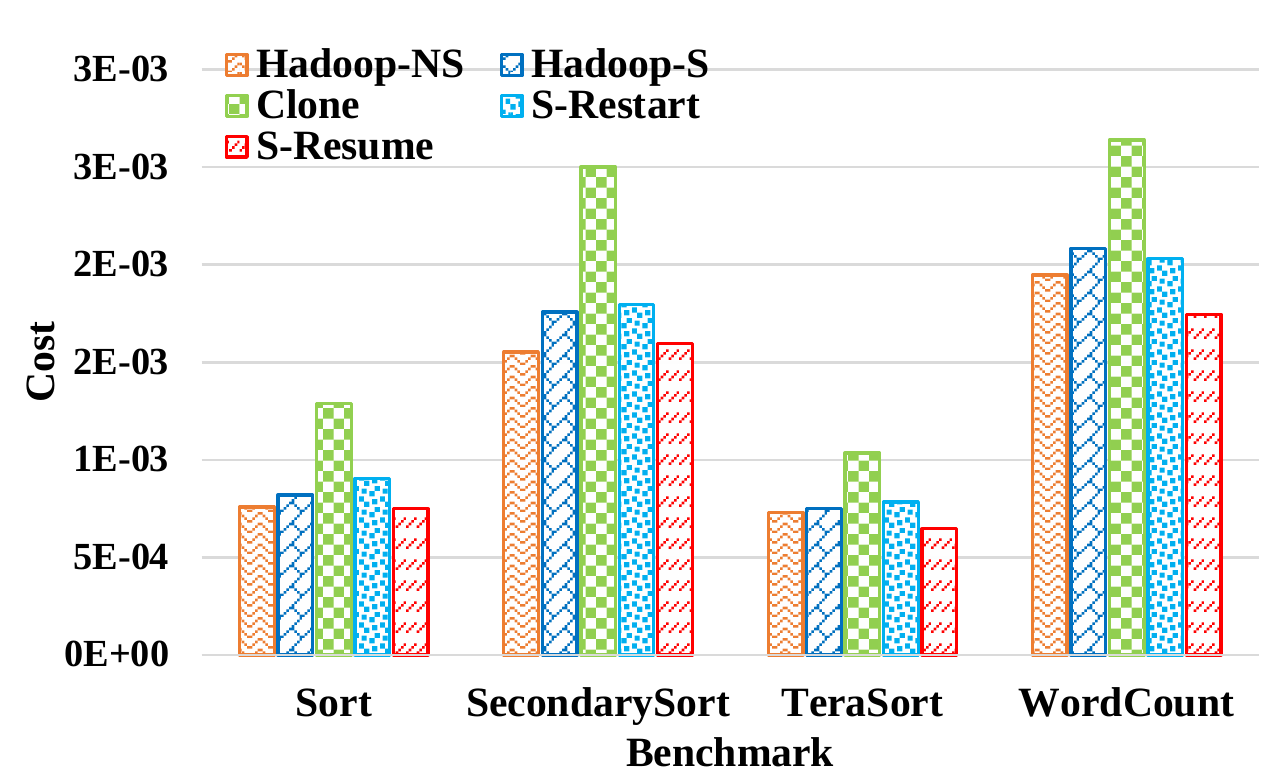}%
		\label{fig:x2}%
	}%
	~
	\subfigure[]{%
		\includegraphics[height=1.8in,width=0.32\textwidth]{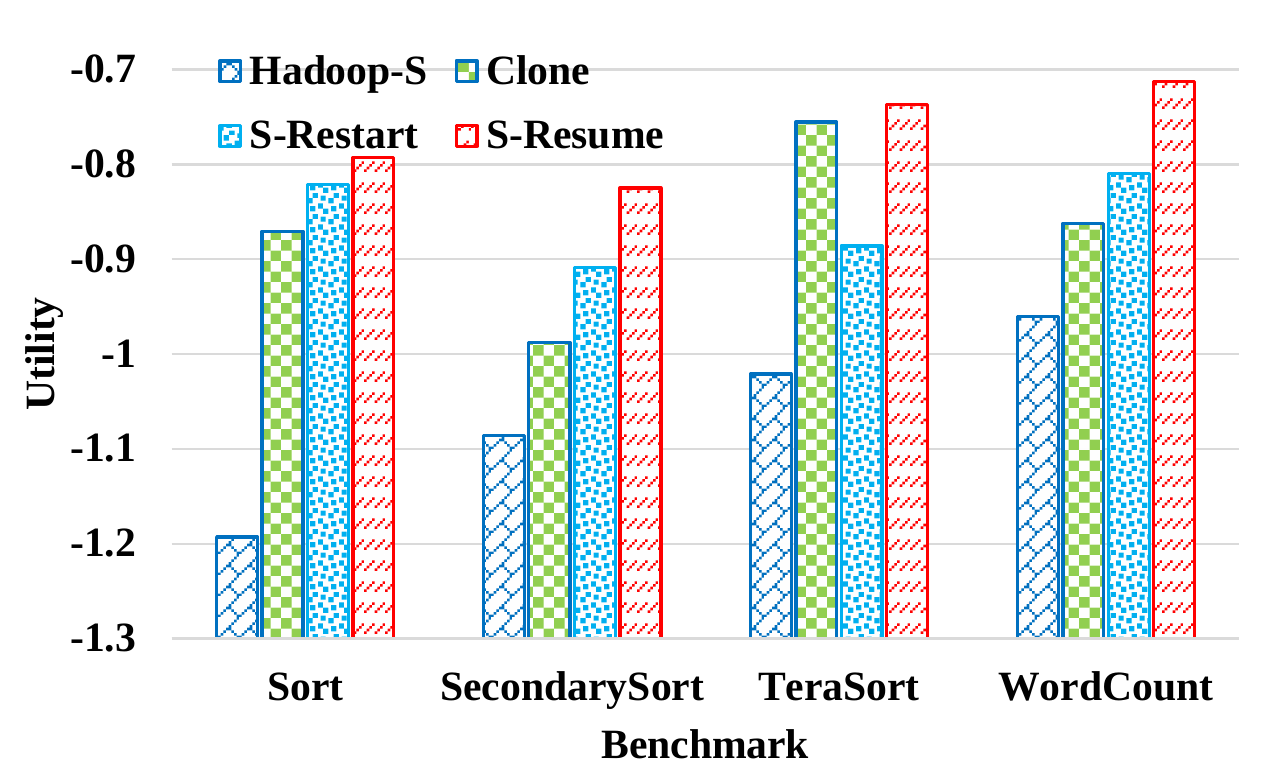}%
		\label{fig:x3}%
	}%
	\vspace{-0.1in}
	\caption{Comparisons of HNS, HS, Clone, S-Restart and S-Resume with respect to PoCD, Cost and Utilities with different benchmarks.}
	\vspace{-0.17in}
\end{figure*}

The task completion time estimation error is mainly caused by Hadoop's assumption that a task starts running right after it is launched. However, this is not true especially in highly contended environments where JVM startup time is significant and cannot be ignored. We therefore take into consideration the time to launch a JVM when estimating task progress and completion time. Task execution time consists of two parts, i.e., the time for launching JVM, and the time for processing the workload. Chronos calculates the time for launching JVM by finding the difference between the time when the first progress report is received ($t_{\rm FP}$) and $t_{\rm lau}$. Therefore, the new estimated completion time is given by
\begin{equation}\label{t_of}
t_{\rm ect} = t_{\rm lau} + (t_{\rm FP}-t_{\rm lau}) + \frac{t_{\rm now}-t_{\rm FP}}{CP-FP},
\end{equation}
where $\frac{t_{\rm now}-t_{\rm FP}}{CP-FP}$ is the time for processing the workload, and $FP$ and $CP$ are the first reported progress value and current reported progress value, respectively.

\subsubsection{Speculative-Restart}
For this strategy, the AM monitors all running tasks and estimates their $t_{\rm ect}$. If the AM finds that a task's estimated completion time $t_{\rm ect}$ at $\tau_{\rm est}$ exceeds deadline $D$, an optimal number $r$ of extra attempts are created and launched. At $\tau_{\rm kill}$, all attempts are killed except for the one with the highest progress score. Unlike default Hadoop, Chronos launches speculative copies once it detects a straggler at $\tau_{\rm est}$.  

\subsubsection{Speculative-Resume}
Our implementation of Speculative-Resume employs a work-preserving mechanism. The key idea is that the AM keeps track of an (original) attempt's data processing progress, maintains the latest (data) record's offset, and passes the offset to the new speculative attempts at $\tau_{\rm est}$ if speculation is determined to be necessary. It ensures that the speculative attempts do not process the already processed data and can seamlessly ``resume" executing data where the original attempt left off.

A challenge in implementing Speculative-Resume is that Chronos needs to consider the launching time of new speculative attempts, since JVM startup time postpones the actual start of execution of new attempts. Our solution  estimates the number of bytes that will be processed ($b_{\rm extra}$) by the original attempt during speculative attempts' launching time. Then, the speculative attempts will skip these bytes and start processing data from an anticipated offset. The speculative launching mechanism allows the original and speculative attempts to switch seamlessly and effectively avoids any JVM startup time overhead. More precisely, at $\tau_{\rm est}$, Chronos records the number of bytes processed by the original attempt ($b_{\rm est}$), and then estimates the number of bytes to be processed by the original attempt during speculative attempt's startup time ($b_{\rm extra}$) as follows:
\begin{equation}\label{imp2}
\frac{b_{\rm est}}{\tau_{\rm est}-t_{\rm FP}}{\cdot}(t_{\rm FP}-t_{\rm lau}).
\end{equation}

Next, our new byte offset to be used by the speculative attempts to start processing data is calculated as $b_{\rm new} = b_{\rm start}+b_{\rm est}+b_{\rm extra}$, where $b_{\rm start}$ is the offset of the first byte of the split, which is the data to be processed by a task. Note that in this strategy, Chronos launches $r + 1$ new copies for each straggler after killing the original one. This is because the original attempt was deemed to be a straggler (for reasons such as node malfunction or operating system disruption), and since all new attempts start from $b_{\rm new}$, there is no data-processing delay for relaunching (and rejuvenating) the original attempt at a new location.

\vspace{-0.1in}
\section{Evaluation}
We evaluate Clone, Speculative Restart (S-Restart), and Speculative Resume (S-Resume) through both testbed experiments and large-scale simulations with real-world data center trace. In all three strategies, the optimal number $r_{\rm opt}$ of clone/speculative attempts are found by solving our proposed PoCD and cost optimization. We compare the three strategies with Default Hadoop without Speculation (Hadoop-NS), Default Hadoop with Speculation (Hadoop-S), and Mantri, which serve as baselines in our evaluation.

\vspace{-0.1in}
\subsection{Experiment Results}

\subsubsection{Setup}
We deploy our prototype of Clone, Speculative Restart, and Speculative Resume strategies on Amazon EC2 cloud testbed consisting of 40 nodes. Each node has 8 vCPUs and 2GB memory.The three strategies are evaluated by using the Map phases of four classic benchmarks, i.e., Sort, SecondarySort, TeraSort, and WordCount. Sort and SecondarySort are I/O bound applications, whereas WordCount and the map phase of TeraSort are CPU-bound. We generate workloads for Sort and TeraSort by using RandomWriter and TeraGen applications, respectively. Also, we generate random number pairs as workload for SecondarySort. The sizes of all workloads are 1.2GB. To emulate a realistic cloud infrastructure with resource contentions, we introduce background applications in the host OS of the physical servers, which repeatedly execute some computational tasks and inject background noise in our testbed. We observe that the task execution time measured on our testbed follows a Pareto distribution with an exponent $\beta<2$. This is consistent with observations made in \cite{grass2014, Xu2016, parallel2015, ren2015hopper}. As shown in Morpheus~\cite{jyothi2016morpheus} and Jockey~\cite{ferguson2012jockey}, job deadlines can be introduced by third parties and/or specified in Service Level Agreements (SLAs). In our experiments, we consider jobs with different deadline sensitivities by setting deadlines as different ratios of average job execution time. To evaluate the impact of background noise (in resource utilization), We use the Stress utility to generate background applications. Stress is a tool to imposes a configurable amount of CPU, memory, I/O, and disk stress on the system, and it is widely used to evaluate cloud and VM performance, such as in~\cite{katherine2012software, nachiyappan2015cloud, carrozza2014evaluation, Stress}.

\subsubsection{Results} In the following, we compare Hadoop-NS, Hadoop-S, Clone, S-Restart, and S-Resume with respect to PoCD, cost, and total net utility. We execute 100 MapReduce jobs using Chronos on our testbed. Each job has 10 tasks and a deadline equal to 100 sec or 150 sec (emulating jobs with different deadline sensitivities). We measure the PoCD by calculating the percentage of jobs that finish before their deadlines, and the cost by multiplying the machine running time, i.e., the average job running time (i.e., VM time required), and a fixed price per unit VM time that is obtained by Amazon EC2 average spot price ($C$). In all experiments, we set the tradeoff factor as $\theta=0.0001$, i.e., 1\% of PoCD utility improvement is equivalent to 100 units of VM cost reduction (since $1\%=\theta\cdot 100$). We then solve the corresponding joint PoCD and cost optimization and compare the performance of various algorithms.

Figure \ref{fig:x1} and Figure \ref{fig:x2} show PoCD and cost of the five strategies with different benchmarks, respectively. The deadline is 100 sec for Sort and TeraSort, and is 150 sec for SecondarySort and WordCount. For our three strategies, $\tau_{\rm est}$ and $\tau_{\rm kill}$ are set to 40 sec and 80 sec, respectively. Hadoop-NS has the lowest PoCD, and relatively large cost. Even though Hadoop-NS does not launch extra attempts, the large execution time of stragglers increases the cost. Hadoop-S can launch extra attempts for slow tasks, but only after at least one task finishes. Because Hadoop-S does not consider deadlines, it might launch extra attempts for tasks that may finish before the deadline. Clone launches extra attempts for each task, and it makes both PoCD and cost be the largest among the five strategies. S-Restart and S-Resume launch extra attempts for tasks whose estimated execution time is larger than the deadline. Due to the fact that S-Resume does not re-process data and kills the original attempt, S-Resume can achieve larger PoCD and smaller cost. Figure \ref{fig:x3} compares the performance of the five strategies in terms of the overall net utility. The results show that our three strategies outperform Hadoop-NS and Hadoop-S by up to $33\%$ on net utility value. In particular, the three strategies can improve PoCD by up to $80\%$ and $19\%$ over Hadoop-NS and Hadoop-S, respectively, while S-Resume introduces little additional cost compared with Hadoop-NS and Hadoop-S. This significant improvement comes from not only launching multiple (optimal number of) attempts for stragglers, but also maintaining only the fastest progress attempt at $\tau_{\rm kill}$, thereby introducing limited execution time overhead. For net utility, since we use the PoCD of Hadoop-NS as $R_{\min}$, its utility is negative infinity.

\begin{table}[ht]
	\caption{Performance comparison with varying $\tau_{\rm est}$ and fixed $\tau_{\rm kill}-\tau_{\rm est} = 0.5{\cdot}t_{\min}$.}
	\label{diff_table}
	\centering
	\begin{tabular}{|c|c|c|c|c|c|}
		\hline
		&$\tau_{\rm est}$           &$\tau_{\rm kill}$          & PoCD & Cost & Utility     \\
		\hline
		Clone                              & 0                     &$0.5{\cdot}t_{\min}$    &0.722  &9373 &-0.376 \\
		\hline
		
		\multirow{3}{1.3cm}{S-Restart}    & $0.1{\cdot}t_{\min}$   &$0.6{\cdot}t_{\min}$    &0.996  &11458 &-0.213          \\
		\cline{2-6}
		& $0.3{\cdot}t_{\min}$   &$0.8{\cdot}t_{\min}$    &0.988  &9650 &-0.199          \\
		\cline{2-6}
		& $0.5{\cdot}t_{\min}$   &$1.0{\cdot}t_{\min}$    &0.938  &9486 &-0.226              \\
		\hline
		\multirow{3}{1.3cm}{S-Resume}     & $0.1{\cdot}t_{\min}$   &$0.6{\cdot}t_{\min}$    &0.997  &9121  &-0.189         \\
		\cline{2-6}
		& $0.3{\cdot}t_{\min}$   &$0.8{\cdot}t_{\min}$    &0.992  &8612  &-0.187          \\
		\cline{2-6}
		& $0.5{\cdot}t_{\min}$   &$1.0{\cdot}t_{\min}$    &0.941  &8712  &-0.217          \\
		\hline
		
	\end{tabular}
	\vspace{-0.15in}
\end{table}

\begin{table}[ht]
	\caption{Performance comparison with varying $\tau_{\rm kill}$ and fixed $\tau_{\rm est}$.}
	\label{est_table}
	\centering
	\begin{tabular}{|c|c|c|c|c|c|}
		\hline
		&$\tau_{\rm est}$           &$\tau_{\rm kill}$          & PoCD & Cost & Utility     \\
		\hline
		\multirow{3}{1.3cm}{Clone}        & 0                     &$0.4{\cdot}t_{\min}$    &0.718  &9113 &-0.376 \\ \cline{2-6}
		& 0                     &$0.6{\cdot}t_{\min}$    &0.733  &9713 &-0.369 \\ \cline{2-6}
		& 0                     &$0.8{\cdot}t_{\min}$    &0.731  &10434 &-0.378 \\
		\hline
		
		\multirow{3}{1.3cm}{S-Restart}    & $0.3{\cdot}t_{\min}$   &$0.4{\cdot}t_{\min}$    &0.981  &8235 &-0.189          \\
		\cline{2-6}
		& $0.3{\cdot}t_{\min}$   &$0.6{\cdot}t_{\min}$    &0.993  &8848 &-0.188          \\
		\cline{2-6}
		& $0.3{\cdot}t_{\min}$   &$0.8{\cdot}t_{\min}$    &0.988  &9650 &-0.199              \\
		\hline
		\multirow{3}{1.3cm}{S-Resume}     &$0.3{\cdot}t_{\min}$    &$0.4{\cdot}t_{\min}$    &0.987  &7935  &-0.183         \\
		\cline{2-6}
		& $0.3{\cdot}t_{\min}$   &$0.6{\cdot}t_{\min}$    &0.993  &8220  &-0.182          \\
		\cline{2-6}
		& $0.3{\cdot}t_{\min}$   &$0.8{\cdot}t_{\min}$    &0.992  &8612  &-0.187          \\
		\hline
		
	\end{tabular}
	\vspace{-0.1in}
\end{table}

\begin{figure*}[t!]
	\centering
	\subfigure[]{%
		\includegraphics[height=1.8in,width=0.33\textwidth]{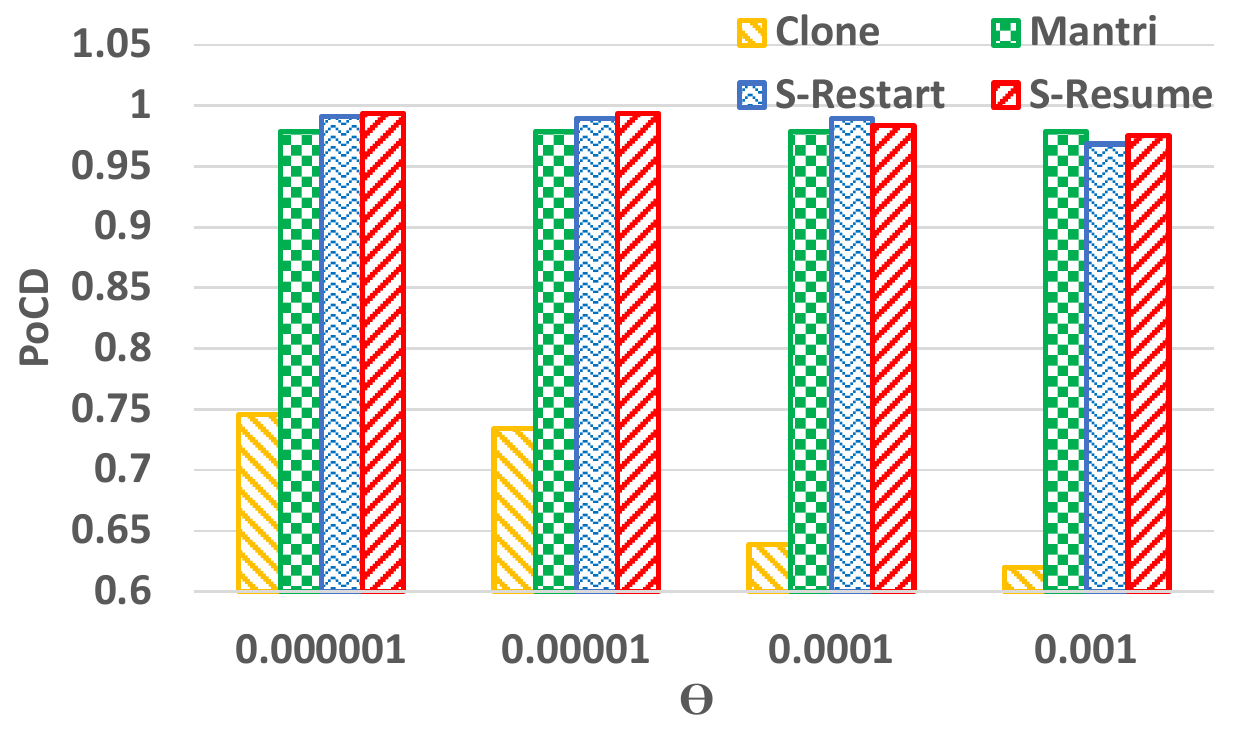}%
		\label{fig:y1}%
	}%
	~
	\subfigure[]{%
		\includegraphics[height=1.8in, width=0.33\textwidth]{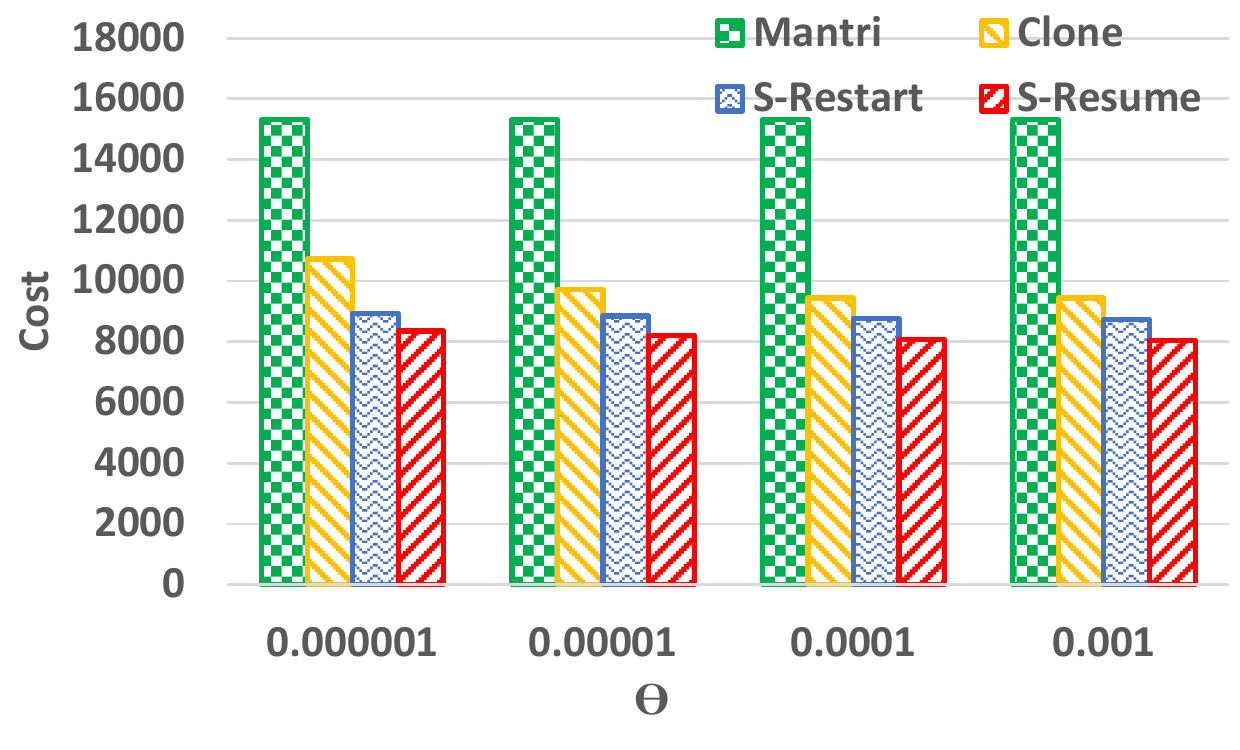}%
		\label{fig:y2}%
	}%
	~
	\subfigure[]{%
		\includegraphics[height=1.8in,width=0.33\textwidth]{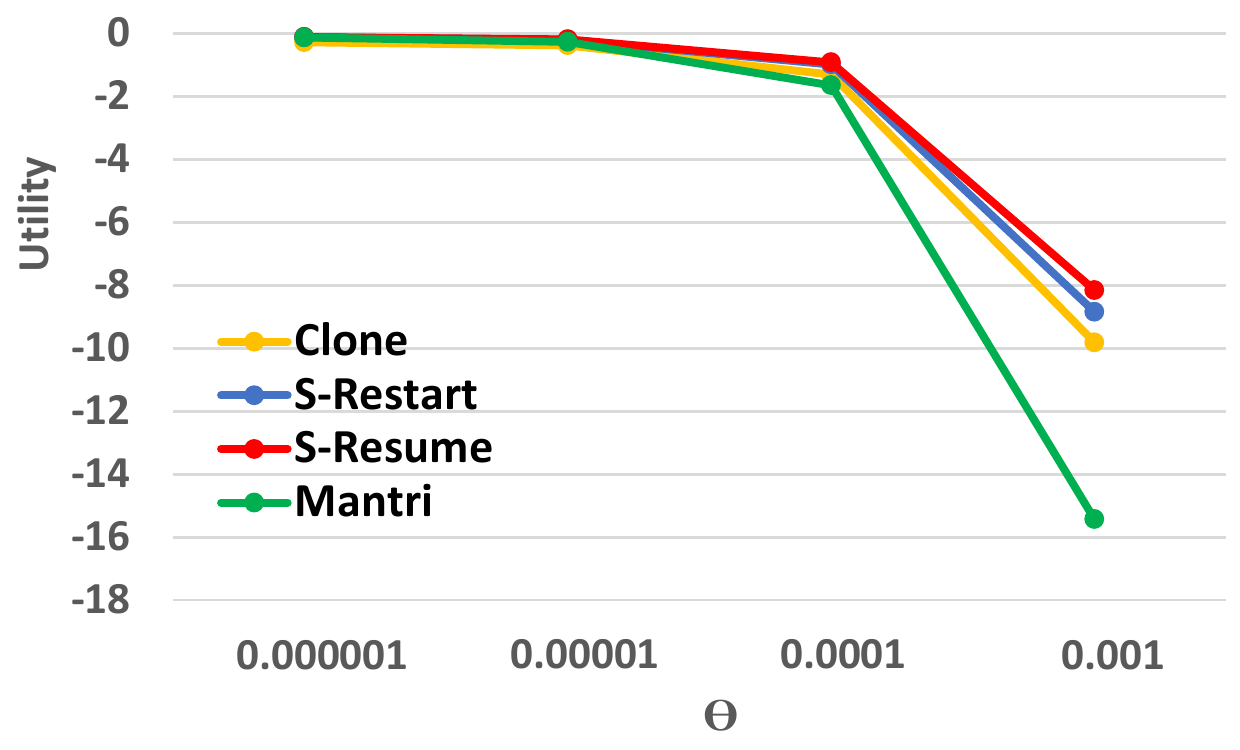}%
		\label{fig:y3}%
	}%
	\vspace{-0.1in}
	\caption{Comparisons of Mantri, Clone, S-Restart and S-Resume on PoCD, Cost, and Utility with different tradeoff factors $\theta$.}
	\vspace{-0.1in}
\end{figure*}

\begin{figure*}[t!]
	\centering
	\subfigure[]{%
		\includegraphics[height=1.8in,width=0.33\textwidth]{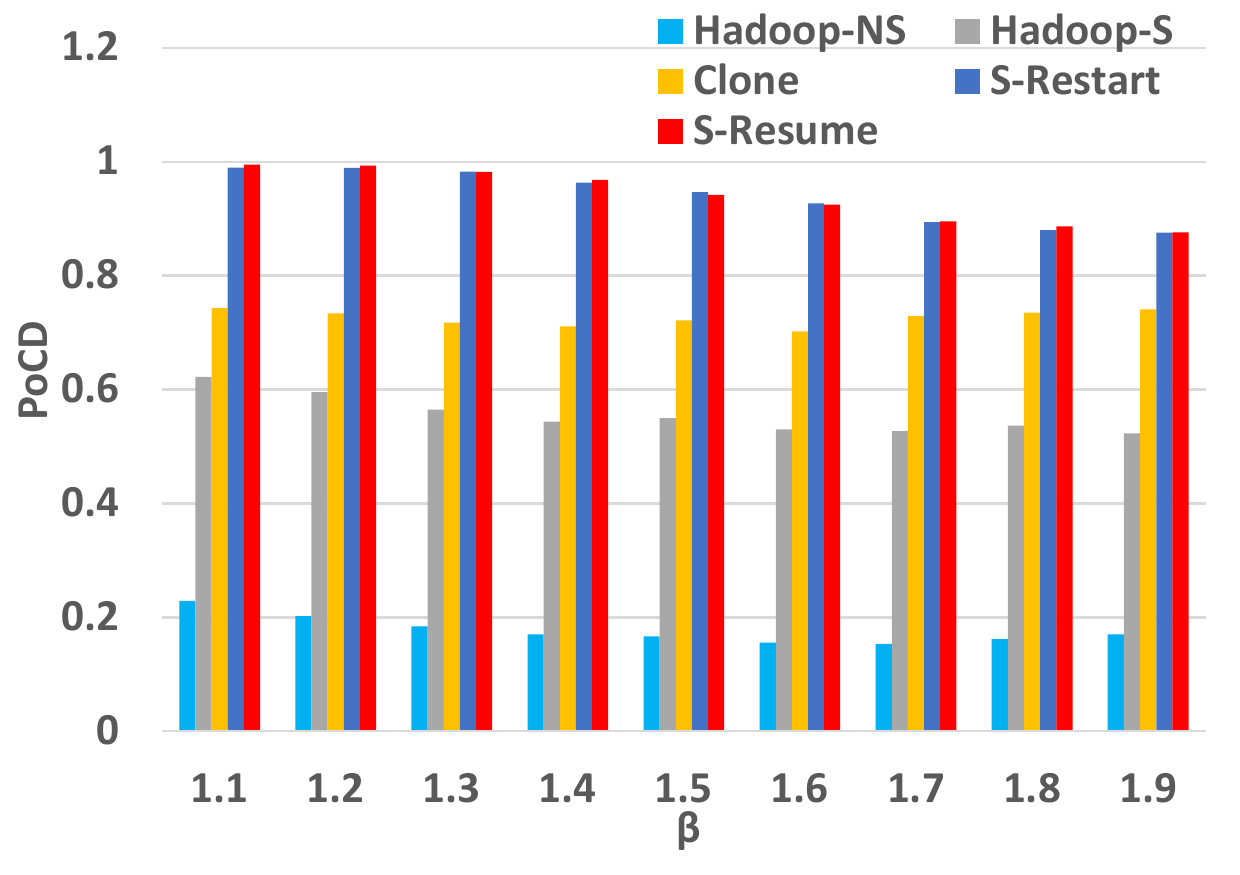}%
		\label{fig:z1}%
	}%
	~
	\subfigure[]{%
		\includegraphics[height=1.8in, width=0.33\textwidth]{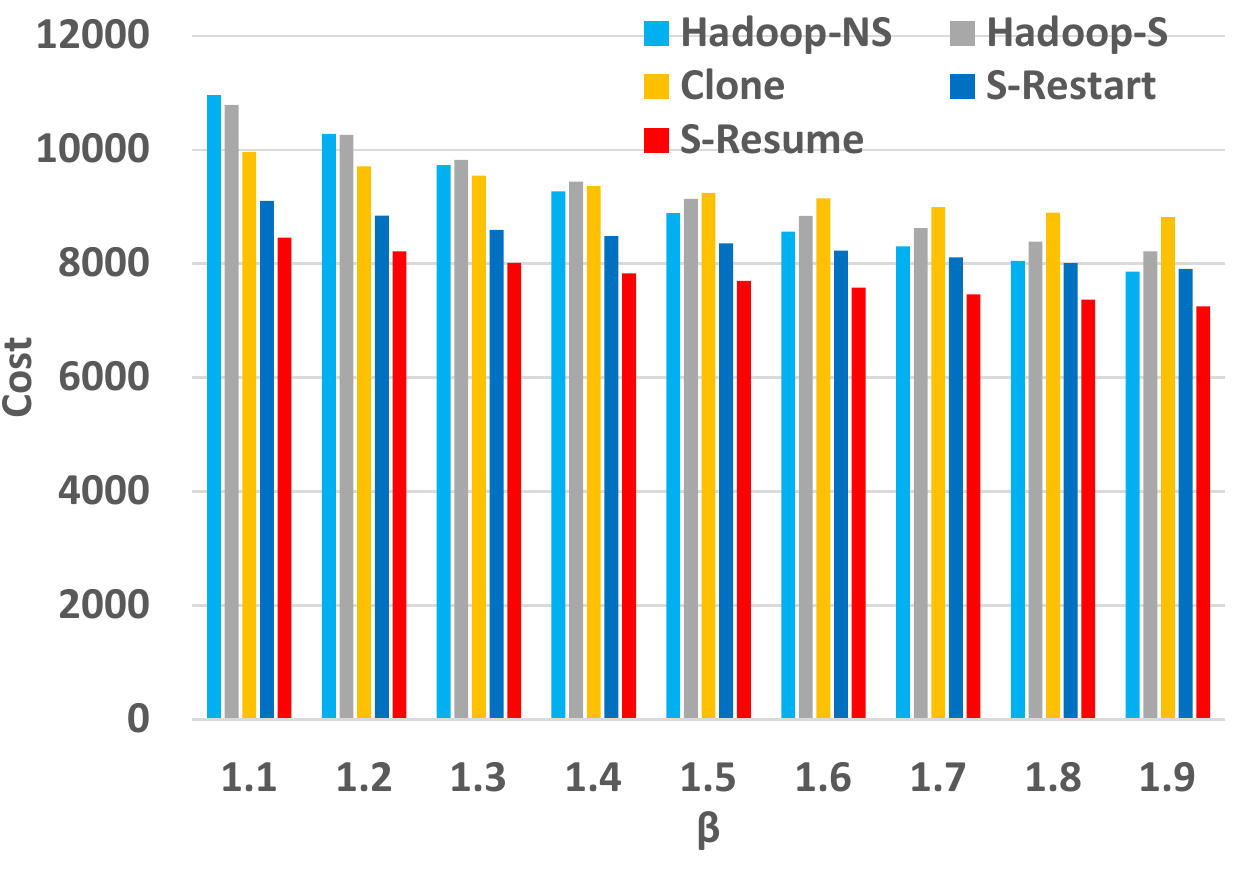}%
		\label{fig:z2}%
	}%
	~
	\subfigure[]{%
		\includegraphics[height=1.8in,width=0.33\textwidth]{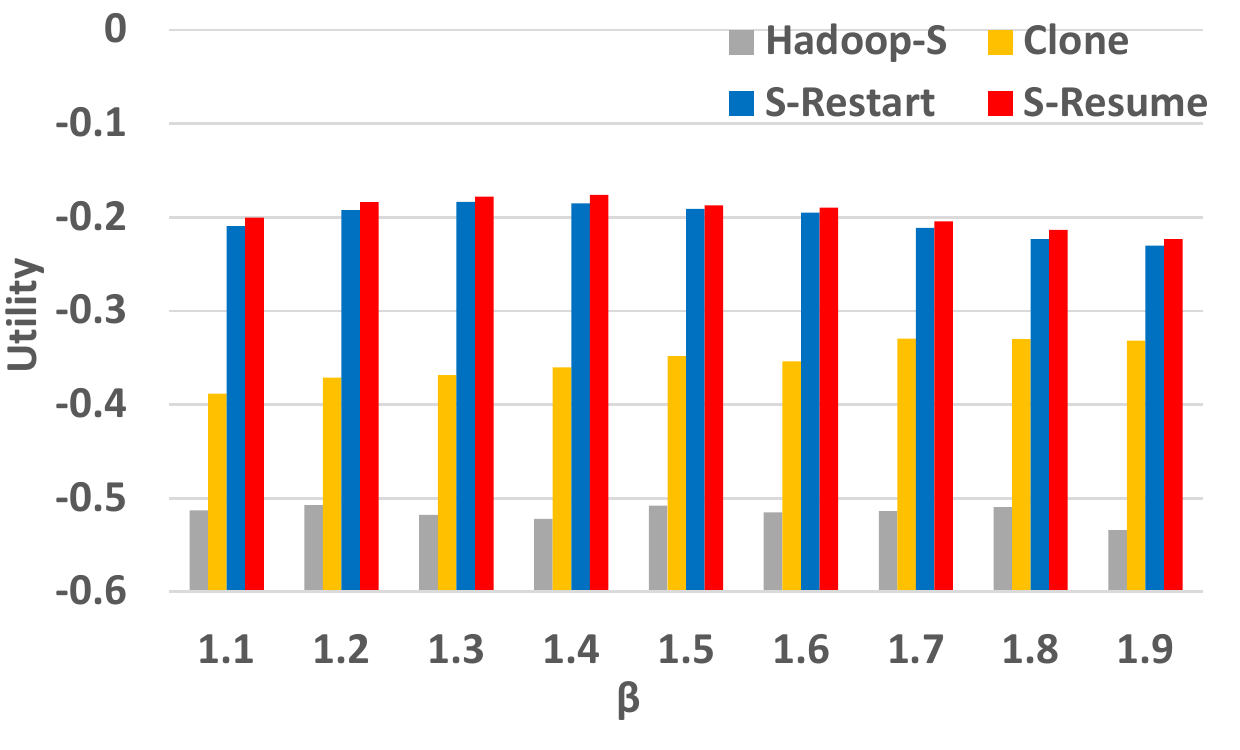}%
		\label{fig:z3}%
	}%
	\vspace{-0.1in}
	\caption{{\color{blue}Comparisons of HNS, HS, Clone, S-Restart and S-Resume in terms of PoCD, MRT and net utilities with different $\beta$s.}}
	\vspace{-0.1in}
\end{figure*}

\subsection{Simulation Results}

Next, to evaluate our proposed Chronos system at scale, we conduct a trace-driven simulation of a large-scale datacenter. We (i) leverage Google trace \cite{GoogleTrace} to generate MapReduce jobs with different parameters and arrival times, and (ii) make use of spot instance price history from Amazon EC2 to calculate the cost of executing different strategies. We simulate $30$ hours of job trace including a total of $2700$ jobs with $1$ million tasks from Google trace. For each job, we get the start time, the execution time distributions, the number of tasks, and job ID, and use Pareto distribution to generate the execution times of tasks that match the execution time distribution in the real-world trace.


To provide insights on determining the optimal parameters (such as $\tau_{\rm est}$ and $\tau_{\rm kill}$) of Chronos, we first study the impact of choosing different parameters in Clone, S-Restart, and S-Resume. Table \ref{diff_table} shows the performance comparison with varying $\tau_{\rm est}$ and fixed difference between $\tau_{\rm est}$ and $\tau_{\rm kill}$, which implies fixed execution time $\tau_{\rm kill}-\tau_{\rm est}$ for all clone/speculative attempts.

Under Clone strategy, $\tau_{\rm est}$ has only one value, which is $0$. Under S-Restart and S-Resume, there exists an interesting tradeoff between estimation accuracy and timeliness of speculation, which is characterized by $\tau_{\rm est}$. Initially, as $\tau_{\rm est}$ increases, both PoCD and cost decrease. This is because with limited observation available for a small $\tau_{\rm est}$, Hadoop tends to overestimate the execution time of attempts at the beginning, while making more accurate estimations as more observations become available for large  $\tau_{\rm est}$. More precisely, when $\tau_{\rm est}$ is small, more tasks are considered as stragglers, and extra attempts are launched more aggressively. This leads to both high PoCD and large cost. On the other hand, when $\tau_{\rm est}$ becomes too large, it might be too late to meet the deadline for straggling tasks.  Based on the table, S-Restart and S-Resume can achieve the highest net utilities when $\tau_{\rm est}=0.3{\cdot}t_{\min}$.


Table \ref{est_table} shows the performance comparisons for varying $\tau_{\rm kill}$ and fixed $\tau_{\rm est}$. We can see that as $\tau_{\rm kill}$ increasing, the execution time of clone/speculative attempts increases, resulting in higher cost, while more accurate estimation of progress and execution time can be obtained. Due to higher cost for executing clone/speculative attempts as $\tau_{\rm kill}$ increases, the value of optimal $r$ decreases, in order to re-balance PoCD and cost in the joint optimization. As a result, PoCDs are not monotonically increasing or decreasing. Based on the results, the best PoCD can be achieved when $\tau_{\rm kill}=0.3{\cdot}t_{\min}$.


In Figure \ref{fig:y1}, \ref{fig:y2}, and \ref{fig:y3}, we compare Mantri, Clone, S-Restart, and S-Resume in terms of PoCD, Cost, and Utility with different $\theta$ values to illustrate the tradeoff achieved by different strategies.

Figure \ref{fig:y1} shows the effect of varying tradeoff factor $\theta$ on PoCD. As $\theta$ increases, the cost has a higher weight and becomes more critical in the joint optimization. Thus, fewer clone/speculative attempts are launched, and PoCD decreases, leading to decreased $r$ to re-balance PoCD and cost in the optimal solution. It is seen that the PoCD of Clone decreases more than others since it incurs higher cost to execute clone attempts than speculative attempts. To validate this, we also present the histograms of the optimal value of $r$ for Clone with $\theta=1e$-$5$, Clone with $\theta=1e$-$4$, S-Resume with $\theta=1e$-$5$, and S-Resume with $\theta=1e$-$4$ in Figure \ref{optimalr}. As $\theta$ increases from $1e$-$5$ to $1e$-$4$, the optimal $r$ for the vast majority of jobs decreases from $r=2$ to $r=1$ under Clone strategy, whereas the optimal majority $r$ decreases from $r=4$ to $r=3$ under S-Resume strategy. Recall that Mantri aggressively launches and kills extra attempts. Thus, with a larger number of extra attempts running along with original tasks, Mantri can achieve high PoCD, but it also runs up significant cost to execute the speculative attempts. Our results show that S-Resume has the highest PoCD and outperforms Clone strategy by $56\%$.

Figure \ref{fig:y2} shows the effect of varying $\theta$ on Cost. As $\theta$ increases, fewer clone/speculative attempts are launched, and the costs of Clone, S-Restart, and S-Resume all decrease. With a larger number of clone/speculative attempts running along with original tasks, the cost of Mantri is much higher than others. The results show that Mantri introduces $50\%$, $67\%$, and $88\%$ higher cost compared with Clone, S-Restart, and S-Resume, respectively, making it non-optimal in terms of net utility optimization.

\begin{figure}
	\centering
	\includegraphics[height=2.0in,width=0.45\textwidth]{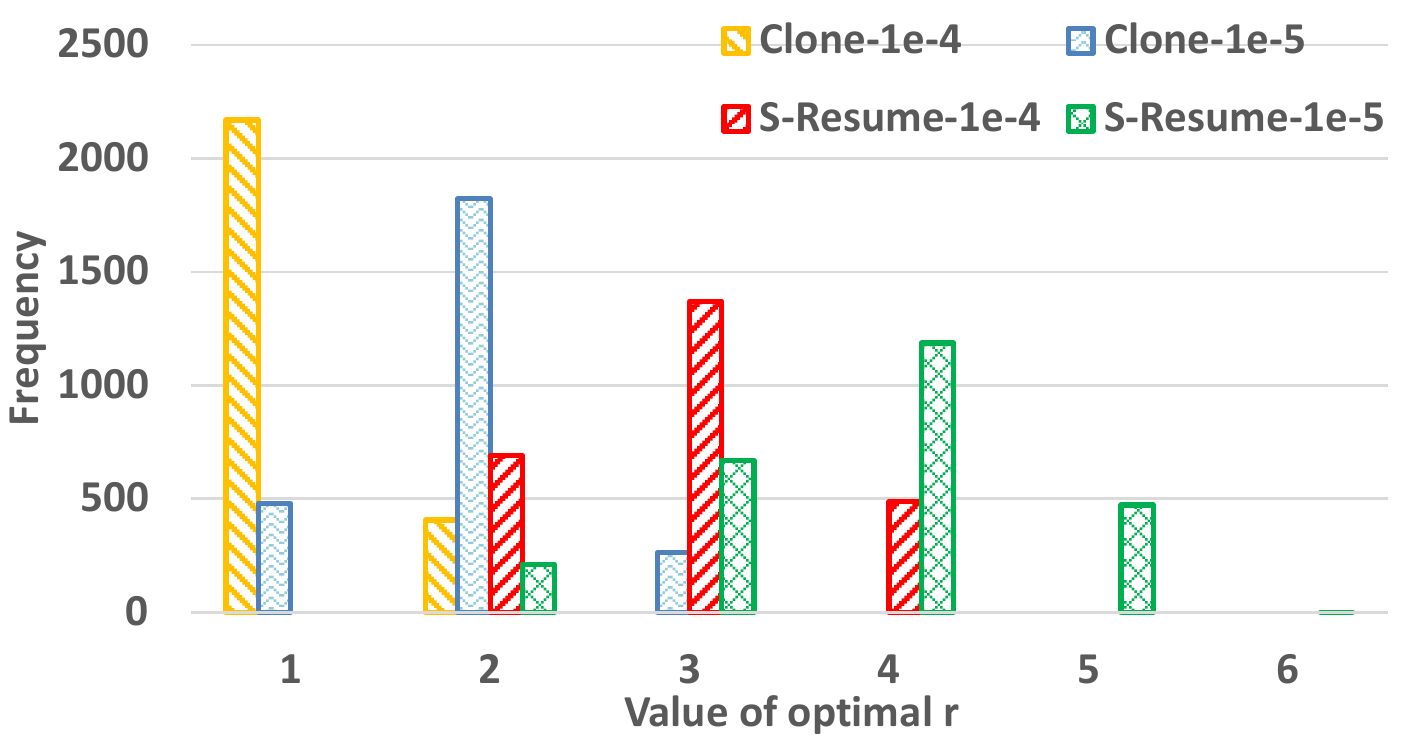}
	\vspace{-0.1in}
	\caption{Histogram of the optimal $r$ values for different strategies}\label{optimalr}
	\vspace{-0.17in}
\end{figure}

Figure \ref{fig:y3} compares the net utility for varying $\theta$. The utility of Mantri decreases the most as $\theta$ increases, since it introduces the largest cost compared with other strategies. S-Resume attains the best net utility, since it has the highest PoCD and lowest cost. Clone, S-Restart, and S-Resume can outperform Mantri by $31\%$, $43\%$, and $50\%$, respectively.

Figure \ref{fig:z1} and \ref{fig:z2} compares the performance of different strategies by varying $\beta$, the exponent in Pareto distribution. A smaller $\beta$ corresponds to a smaller decaying in task execution time distribution, implying a heavier tail. In this simulation, we set deadline as $2$ times the average task execution time. Because average execution time is $t_{\min}+\frac{t_{\min}}{\beta-1}$ under Pareto distribution, as $\beta$ increases, the average execution time decreases, which results in decreasing costs. Also, as $\beta$ increasing, the value of optimal $r$ decreases due to a smaller tail in task execution time.

Figure \ref{fig:z3} shows net utility achieved by different strategies under various $\beta$. We observe that Clone, S-Restart and S-Resume outperform Hadoop-NS and Hadoop-S when $\beta$ varies from $1.1$ to $1.9$.

In summary, as our analysis and experiment results show, if the unit cost of machine running time decreases, Clone becomes more preferable than S-Resume due to smaller up-front cost that is resulted from running the clone attempts (starting from $t=0$). Further, S-Resume always outperforms S-Restart, since it is able to capitalize on the partially computed results and pick up the task execution from the breakpoint. But we also note that S-Resume may not be possible in certain (extreme) scenarios such as system breakdown or VM crash, where only S-Restart is feasible.

\vspace{-0.08in}
\section{Conclusion}
In this paper, we present Chronos, a unifying optimization framework to provide probabilistic guarantees for jobs with deadlines. Our framework aims to maximize the probability of meeting job deadlines and minimize the execution cost. We present and optimize three different strategies to mitigate stragglers proactively and reactively. Our solution includes an algorithm to search for the optimal number of speculative copies needed. Moreover, we present an improved technique to estimate the expected task completion time. Our results show that Chronos can outperform by up to $80\%$ in terms of PoCD over Hadoop-NS, and by $88\%$ in terms of cost over Mantri. Multi-wave executions will be considered in our future work.




%

\bibliographystyle{IEEEtran}
\bibliography{references}

\section{Appendix}\label{appendix}
In this section, we present details for proving Theorem \ref{restart_PoCD}, \ref{restart_time}, \ref{resume_PoCD}, \ref{resume_time}, \ref{compare}, \ref{opt_clone}. To prove {\em Theorem} \ref{restart_time}. To prove {\em Theorem} \ref{restart_time}, we first present {\em Lemma} \ref{lemma2}.

\begin{lemma}\label{lemma2}
	Let $A$ and $B$ follow Pareto distribution with $(a_{\min},\beta)$ and $(b_{\min},\beta)$, respectively. Given $a_{\min}{\leq}b_{\min}$, then
	\begin{equation}\label{lemma2_1}
	f_B(b) = f_A(b|A>b_{\min})
	\end{equation}
\end{lemma}

\noindent{\textbf{PROOF OF THEOREM \ref{restart_PoCD}}}

\begin{proof}
	In this proof, we first derive the probability of a task completes before the deadline, and then derive the PoCD by considering all $N$ tasks completing before the deadline.
	
	We first compute the probability that an original attempt fails to finish before $D$. We denote the probability that an original attempt fails to finish before $D$ as $P_{\rm Restart,o}$, which equals
	\begin{equation}\label{restart1}
	P_{\rm Restart,o}=P(T_{j,1}>D)=\left(\frac{t_{\min}}{D}\right)^{\beta}.
	\end{equation}
	
	If the original attempt's execution time is larger than $D$, $r$ extra attempts are launched at $\tau_{\rm est}$. If an extra attempt fails to finish before $D$, it means that the execution time is more than $D-\tau_{\rm est}$. We denote the probability that an extra attempt fails to finish before $D$ as $P_{\rm Restart,e}$, which equals
	\begin{equation}\label{restart2}
	P_{\rm Restart,e} = P(T_{j,a}>D-\tau_{\rm est})
	=\int_{D-{\tau_{\rm est}}}^{\infty}{\frac{{\beta}{t^{\beta}_{\min}}}{t^{\beta+1}}}{dt}
	=\left(\frac{t_{\min}}{D-{\tau_{\rm est}}}\right)^{\beta}.
	\end{equation}
	
	A task fails to finish before $D$ when the original attempt and $r$ extra attempts launched at $\tau_{est}$ fail to finish before $D$. Thus, the probability that a task can finish before $D$ is $1-P_{\rm Restart,o}{\cdot}P^r_{\rm Restart,e}$. Also, a job can finish before the deadline $D$ when all N tasks can finish before $D$. Thus, PoCD $R_{\rm S-Restart}$ equals
	\begin{equation}\label{restart3}
	(1-P_{\rm restart,o}{\cdot}(P_{\rm restart,e})^r)^N = \left[1-\frac{t^{\beta{\cdot}(r+1)}_{\min}}{D^{\beta}{\cdot}(D-\tau_{\rm est})^{\beta{\cdot}r}}\right]^N.
	\end{equation}
\end{proof}

\noindent{\textbf{PROOF OF THEOREM \ref{restart_time}}}

\begin{proof}
	In this proof, we first derive machine running time of a task by considering if the execution time of the original attempt is larger than $D$. If execution time is no more than $D$, there is no extra attempt launched. The machine running time is execution time of original attempt. If the execution time is larger than $D$, task machine time is the summation of machine running time of extra attempts killed at $\tau_{kill}$ and machine running time of the attempt successfully completed, and then get job's machine running time by adding machine running time of $N$ tasks together.
	
	$E_{\rm S-Restart}(T)$ equals the expectation machine running time of all $N$ tasks, i.e., $E_{\rm S-Restart}(T) = N{\cdot}E(T_j)$, where $T_j$ is the machine running time of task $j$. For task $j$, if the execution time of the original attempt exceeds $D$ ($T_{j,1}{>}D$), extra attempts are launched at $\tau_{est}$. Thus, we consider $E(T_j)$ in whether the execution time of the original attempt exceeds $D$ or not, i.e., $T_{j,1}{\leq}D$, or $T_{j,1}{>}D$, and $E(T_j)$ equals
	\begin{equation}\label{restart5}
	E(T_j|T_{j,1}{\leq}D){\cdot}P(T_{j,1}{\leq}D)+E(T_j|T_{j,1}{>}D){\cdot}P(T_{j,1}{>}D),
	\end{equation}
	where
	\begin{equation}\label{restart6}
	P(T_{j,1}{>}D) = 1-P(T_{j,1}{\leq}D)=\left(\frac{t_{\min}}{D}\right)^{\beta}.
	\end{equation}
	
	\vspace{3mm}
	
	{\em \textbf{Case 1. $T_{j,1}{\leq}D$.}}
	
	\vspace{3mm}
	
	The execution time of the original attempt is no more than $D$, and there is no extra attempt launched at time $\tau_{est}$. Because
	\begin{equation}\label{restart8}
	E(T_j|T_{j,1}{\leq}D)=\int_{t_{\min}}^{\infty}t{\cdot}f_{T_j}(t|T_{j,1}{\leq}D)dt,
	\end{equation}
	we first compute $f_{T_j}(t|T_{j,1}{\leq}D)$, which is
	\begin{eqnarray}\label{restart14}
	f_{T_j}(t|T_{j,1}{\leq}D)=\left\{
	\begin{array}{lll}
	\frac{\beta{\cdot}t^{\beta}_{\min}{\cdot}D^{\beta}}{(D^{\beta}-t^{\beta}_{\min}){\cdot}t^{\beta+1}}, \ \ D{>}t{>}t_{\min};\\
	\ \ \ \ \ 0, \qquad \qquad \ \ otherwise,
	\end{array}
	\right.
	\end{eqnarray}
	By deriving Equ. (\ref{restart8}), we can get
	\begin{equation}\label{restart7}
	E(T_j|T_{j,1}{\leq}D) =
	\frac{t_{\min}{\cdot}D{\cdot}\beta{\cdot}(t^{\beta-1}_{\min}-D^{\beta-1})}{(1-\beta){\cdot}(D^{\beta}-t^{\beta}_{\min})}.
	\end{equation}
	
	\vspace{3mm}
	
	{\em \textbf{Case 2. $T_{\rm j,1}{>}D$. }}
	
	\vspace{3mm}
	
	The execution time of the original attempt is more than $D$, and $r$ extra attempts are launched at time $\tau_{\rm est}$. At time $\tau_{\rm kill}$, a attempt with the smallest finish time is left for running, and other $r$ attempts are killed.
	
	$E(T_j|T_{j,1}{>}D)$ consists of three parts, i.e., machine running time of the original attempt before $\tau_{\rm est}$, machine running time of $r$ killed attempts between $\tau_{\rm est}$ and $\tau_{kill}$, and execution time of the attempt with the smallest finish time after $\tau_{\rm est}$. Since the original attempt starts $\tau_{est}$ earlier than extra attempts, $T_{j,1}-{\tau_{\rm est}}$ is the execution time of the original attempt after $\tau_{\rm est}$, and $\min(T_{j,1}-{\tau_{\rm est}},T_{j,2},...,T_{j,r+1})$ is the execution time of the attempt with the smallest finish time after $\tau_{\rm est}$. We denote $\min(T_{j,1}-{\tau_{est}},T_{j,2},...,T_{j,r+1})$ as ${W}^{all}_{j}$. So, $E(T_j|T_{j,1}{>}D)$ equals
	\begin{equation}\label{restart12}
	\tau_{est}+r{\cdot}(\tau_{\rm kill}-\tau_{\rm est})+E({W}^{all}_{j}|T_{\rm j,1}{>}D).
	\end{equation}

	Based on Lemma \ref{lemma2}, we can transform  $E({W}^{\rm all}_{j}|T_{j,1}{>}D)$ to $E(\widehat{W}^{\rm all}_{j})$, where $\widehat{W}^{\rm all}_{j}=\min(\widehat{T}_{j,1}-\tau_{\rm est},{T}_{j,2},..,{T}_{j,r+1})$, and the minimum values of $\widehat{T}_{j,1}$ is $D$. Also, we denote $\min({T}_{j,2},...,{T}_{j,r+1})$ as ${W}^{\rm extra}_{j}$.
	
	We compute $E(\widehat{W}^{\rm all}_{j})$ by using {\em Lemma} \ref{lemma1}, i.e.,
	\begin{equation}\label{restart13}
	E(\widehat{W}^{\rm all}_{j})=\int_{t_{\min}}^{\infty}P(\widehat{T}_{j,1}-\tau_{\rm est}{\geq}\omega)
	{\cdot}P^r({T}_{\rm j,a}{\geq}\omega)d\omega+t_{\min},
	\end{equation}
	where
	\begin{eqnarray}\label{restart14}
	P(\widehat{T}_{j,1}-\tau_{\rm est}{\geq}\omega)=\left\{
	\begin{array}{lll}
	(\frac{D}{\omega+\tau_{est}})^{\beta}, \ \ \omega{\geq}D-\tau_{\rm est};\\
	\ \ \ \ \ 1, \ \ \ \ \ \ \ \omega{<}D-\tau_{\rm est},
	\end{array}
	\right.
	\end{eqnarray}
	and
	\begin{eqnarray}\label{restart15}
	P({T}_{j,a}{\geq}\omega)=\left\{
	\begin{array}{lll}
	(\frac{t_{\min}}{\omega})^{\beta}, \ \ \omega{\geq}t_{\min};\\
	\ \ \ \ \ 1, \ \ \ \ \ \omega{<}t_{\min}.
	\end{array}
	\right.
	\end{eqnarray}
	
	Because $D-\tau_{est}$ should be no less than $t_{min}$, otherwise there is no reason for launching extra attempts, thus
	\begin{multline}\label{restart16}
	E(\widehat{W}^{all}_{j}) = \frac{t_{\min}}{\beta{\cdot}r-1}-\frac{t^{\beta{\cdot}r}_{\min}}{(\beta{\cdot}r-1){\cdot}(D-\tau_{\rm est})^{\beta{\cdot}r-1}} \\
	+\int_{D-\tau_{\rm est}}^{\infty}\left(\frac{D}{\omega+\tau_{est}}\right)^{\beta}{\cdot}\left(\frac{t_{\min}}{\omega}\right)^{\beta{\cdot}r}d\omega
	+t_{\min}
	\end{multline}
\end{proof}

\noindent{\textbf{PROOF OF THEOREM \ref{resume_PoCD}}}

\begin{proof}
	In this proof, we first derive the probability that a task completes before the deadline, and then derive the PoCD by considering all $N$ tasks complete before the deadline.
	
	We first compute the probability of an original attempt failed to finish before $D$. We denote the probability that an original attempt fails to finish before $D$ as $P_{\rm Resume,o}$ that
	\begin{equation}\label{resume1}
	P_{\rm Resume,o}=P(T_{j,1}>D)=\left(\frac{t_{\min}}{D}\right)^{\beta}.
	\end{equation}
	
	If the original attempt's execution time is larger than $D$, the original attempt is killed, and $r+1$ extra attempts are launched at $\tau_{\rm est}$. Suppose the original attempt processed $\varphi_{j,\rm est}$ fraction of data, extra attempts continue to process the remaining $1-\varphi_{j,\rm est}$ fraction of data. An extra attempt failed to finish before $D$ means the execution time is more than $D-\tau_{\rm est}$. We denote the probability that an extra attempt fails to finish before $D$ as $P_{\rm Resume,e}$ that
	\begin{equation}\label{resume2}
	P_{\rm Resume,e}=P((1-\varphi_{j,\rm est}){\cdot}T_{j,a}>D-\tau_{\rm est})=\left[\frac{(1-\varphi_{j,\rm est}){\cdot}t_{\min}}{D-\tau_{\rm est}}\right]^{\beta}.
	\end{equation}
	
	A task fails to finish before $D$ when the original attempt and the $r$ extra attempts launched at $\tau_{\rm est}$ fail to finish before $D$. Thus, the probability that a task finishes before $D$ is $1-P_{\rm Resume,o}{\cdot}P^{r+1}_{\rm Resume,e}$. Also, the job finishes before the deadline $D$ when all $N$ tasks finish before $D$. Thus, PoCD $R_{\rm S-Resume}$ equals
	\begin{multline}\label{resume3}
	R_{\rm S-Resume}=(1-P_{\rm resume,o}{\cdot}(P_{\rm resume,e})^{r+1})^N \\
	= \left[1-\frac{(1-\varphi_{\rm j,est})^{{\beta}{\cdot}(r+1)}{\cdot}t^{{\beta}\cdot(r+2)}_{\min}}{D^{\beta}{\cdot}(D-\tau_{\rm est})^{{\beta}{\cdot}(r+1)}}\right]^N.
	\end{multline}
\end{proof}

\noindent{\textbf{PROOF OF THEOREM \ref{resume_time}}}

\begin{proof}
	In this proof, we first derive machine running time of a task by considering if the execution time of the original attempt is larger than $D$. If execution time is no more than $D$, there is no extra attempt launched. The machine running time is execution time of original attempt. If the execution time is larger than $D$, task machine time is the summation of machine running time of extra attempts killed at $\tau_{\rm kill}$ and machine running time of the attempt successfully completed, and then get job's machine running time by adding machine running time of $N$ tasks together.
	
	$E_{S-Resume}(T)$ equals the expectation machine running time of all $N$ tasks, i.e., $E_{S-Resume}(T) = N{\cdot}E(T_j)$, where $T_j$ is the machine running time of task $j$. For task $j$, if the execution time of the original attempt exceeds $D$ ($T_{j,1}{>}D$), extra attempts are launched at $\tau_{\rm est}$. Thus, we consider $E(T_j)$ in whether the execution time of the original attempt exceeds $D$ or not, i.e., $T_{\rm j,1}{\leq}D$, or $T_{\rm j,1}{>}D$, and $E(T_j)$ equals
	\begin{equation}\label{resume4}
	E(T_j|T_{\rm j,1}{\leq}D){\cdot}P(T_{\rm j,1}{\leq}D)+E(T_j|T_{\rm j,1}{>}D){\cdot}P(T_{\rm j,1}{>}D),
	\end{equation}
	where
	\begin{equation}\label{resume5}
	P(T_{\rm j,1}{>}D) = 1-P(T_{j,1}{\leq}D)=\left(\frac{t_{\min}}{D}\right)^{\beta}.
	\end{equation}
	
	\vspace{3mm}
	
	{\em \textbf{Case 1. $T_{\rm j,1}{\leq}D$.}}
	
	\vspace{3mm}
	
	The execution time of the original attempt is no more than $D$, and there is no extra attempt launched at time $\tau_{est}$. Because
	\begin{equation}\label{resume_4}
	E(T_j|T_{\rm j,1}{\leq}D)=\int_{t_{\min}}^{\infty}t{\cdot}f_{T_j}(t|T_{j,1}{\leq}D)dt,
	\end{equation}
	we first compute $f_{T_j}(t|T_{j,1}{\leq}D)$, which is
	\begin{eqnarray}\label{resume_5}
	f_{T_j}(t|T_{j,1}{\leq}D)=\left\{
	\begin{array}{lll}
	\frac{\beta{\cdot}t^{\beta}_{\min}{\cdot}D^{\beta}}{(D^{\beta}-t^{\beta}_{\min}){\cdot}t^{\beta+1}}, \ \ D{>}t{>}t_{\min};\\
	\ \ \ \ \ 0, \qquad \qquad \ \ otherwise,
	\end{array}
	\right.
	\end{eqnarray}
	By deriving Equ. (\ref{resume_4}), we can get
	\begin{equation}\label{resume_6}
	E(T_j|T_{j,1}{\leq}D) =
	\frac{t_{\min}{\cdot}D{\cdot}\beta{\cdot}(t^{\beta-1}_{\min}-D^{\beta-1})}{(1-\beta){\cdot}(D^{\beta}-t^{\beta}_{\min})}.
	\end{equation}
	
	\vspace{3mm}
	
	{\em \textbf{Case 2 $T_{\rm j,1}{>}D$}}
	
	\vspace{3mm}
	
	The execution time of the original attempt is more than $D$, and the original attempt is killed at $\tau_{\rm est}$. $r+1$ extra attempts are launched at $\tau_{est}$, and continue to process remaining $(1-\varphi_{j,\rm est})$ data. At $\tau_{kill}$, $r$ extra attempts are killed and leave an extra attempt with the smallest finishing time running. $E(T_j|T_{j,1}{>}D)$ consists of three parts, i.e., machine running time of the original attempt before $\tau_{\rm est}$, machine running time of killed extra attempts between $\tau_{\rm est}$ and $\tau_{\rm kill}$, and execution time of the extra attempt with the smallest finish time. So,
	\begin{equation}\label{resume7}
	E(T_j|T_{j,1}{>}D) = \tau_{\rm est}+r{\cdot}(\tau_{kill}-\tau_{\rm est})+E(W^{new}_j),
	\end{equation}
	where $W^{new}_j=\min(T_{j,2},...,T_{j,r+1})$, and $T_{j,a}=(1-\varphi_{j,\rm est}){\cdot}T_{j,1}$, $\forall a\in{2,...,r+1}$. Based on Lemma \ref{lemma1},
	\begin{align}\label{resume8}
	E(W^{new}_j) &=\int_{t_{\min}}^{\infty}P^{r+1}((1-\varphi_{j,\rm est}){\cdot}T_{j,1}>t)dt+t_{\min} \\
	\ &=\frac{t_{\min}{\cdot}(1-\varphi_{j,\rm est})^{\beta{\cdot}(r+1)}}{{\beta}{\cdot}(r+1)-1}+t_{\min}.
	\end{align}
\end{proof}

\noindent{\textbf{PROOF OF THEOREM \ref{compare}}}

\begin{proof}
	Suppose function $f(x)$ equals $(1-x)^N$, where $N$ is an integer. $f(x)$ is a monotonic non-increasing function, i.e., $f(x_1){\geq}f(x_2)$, if and only if $x_1{\leq}x_2$.
	
	In following, we compare among $(1-R_{\rm clone})^{1/N}$, $(1-R_{\rm S-Restart})^{1/N}$, and $(1-R_{\rm S-Resume})^{1/N}$.
	\subsubsection{PoCD comparison between Clone and Speculative-Restart}
	\begin{equation}\label{com_4}
	\frac{(1-R_{\rm clone})^{1/N}}{(1-R_{\rm S-Restart})^{1/N}} = \left(\frac{D-\tau_{\rm est}}{D}\right)^{\beta{\cdot}r} < 1
	\end{equation}
	Thus, $R_{\rm clone}>R_{\rm S-Restart}$.
	\subsubsection{PoCD comparison between Speculative-Restart and Speculative-Resume}
	\begin{equation}\label{com_5}
	\frac{(1-R_{\rm S-Restart})^{1/N}}{(1-R_{\rm S-Resume})^{1/N}}=\frac{(D-\tau_{\rm est})^{\beta}}
	{t^{\beta}_{\min}{\cdot}(1-\varphi_{j,\rm est})^{\beta{\cdot}(r+1)}}
	\end{equation}
	Given $D-\tau_{\rm est}{\geq}t_{\min}{\cdot}(1-\varphi_{j,\rm est})$, we can get $R_{\rm S-Restart}{<}R_{\rm S-Resume}$.
	\subsubsection{PoCD comparison between Clone and Speculative-Resume}
	\begin{equation}\label{com_6}
	\frac{(1-R_{\rm clone})^{1/N}}{(1-R_{\rm S-Resume})^{1/N}}=\frac{(D-\tau_{\rm est})^{\beta{\cdot}(r+1)}}
	{(1-\varphi_{j,\rm est})^{\beta{\cdot}(r+1)}{\cdot}D^{\beta{\cdot}r}{\cdot}t^{\beta}_{\min}}
	\end{equation}
	Given the original attempt misses the deadline, then $D-\tau_{est}<(1-\varphi_{j,\rm est}){\cdot}D$. Thus, if
	\begin{equation}\label{com_7}
	r>\log_{\frac{D-\tau_{\rm est}}{(1-\varphi_{j, \rm est}){\cdot}D}}\frac{(1-\varphi_{j,\rm est})^{\beta}{\cdot}t^{\beta}_{\min}}{D-\tau_{\rm est}},
	\end{equation}
	then $R_{\rm clone}{>}R_{\rm S-Resume}$. Otherwise, $R_{\rm clone}{\leq}R_{\rm S-Resume}$.
\end{proof}

\noindent{\textbf{PROOF OF THEOREM \ref{opt_clone}}}

\begin{proof}
	In following, we prove $\lg(R_{\rm Clone}(r)-R_{\min})$ is a concave function when $r{>}\Gamma_{\rm Clone} =-\beta^{-1}{\cdot}\log_{t_{\min}/D}{N}-1$, and $-\theta{\cdot}C{\cdot}E_{\rm Clone}(T)$ is a concave function for all values of $r$.
	
	$\lg(R_{\rm Clone}(r)-R_{\min})$ is an increasing and concave function of $R_{\rm Clone}(r)$. Also, when $r{>}\Gamma_{\rm Clone}$, the second order derivative of $R_{\rm Clone}$ is less than $0$. So, $R_{\rm Clone}$ is a concave function of $r$, when $r{>}\Gamma_{\rm Clone}$. Based on {\em Lemma} \ref{lemma3}, we know that $\lg(R_{\rm Clone}-R_{\min})$ is a concave function of $r$, when $r{>}\Gamma_{\rm Clone}$.
	
	The second order derivative of $-E_{Clone}(T)$ equals
	\begin{equation}
	-\frac{2{\cdot}N{\cdot}t_{\min}}{[\beta{\cdot}(r+1)-1]^{3}},
	\end{equation}
	which is less than $0$. Given the summation of two concave functions is also a concave function, thus $U_{\rm Clone}(r)$ is a concave function when $r{>}\Gamma_{\rm Clone}$.
\end{proof}

\begin{proof}
	In following, we prove $\lg(R_{\rm S-Restart}-R_{\min})$ is a concave function when
	\begin{equation}
	r > \Gamma_{\rm S-Restart} = \beta^{-1}{\cdot}\log_{t_{\min}/(D-\tau_{\rm est})}\frac{D^{\beta}}{N{\cdot}t_{\min}},
	\end{equation}
	and $-\theta{\cdot}C{\cdot}E_{\rm S-Restart}(T)$ is a concave function for all values of $r$.
	
	$\lg(R_{\rm S-Restart}(r)-R_{\min})$ is an increasing and concave function of $R_{\rm S-Restart}(r)$. Also, when $r{>}\Gamma_{\rm S-Restart}$, the second order derivative of $R_{\rm S-Restart}$ is less than $0$. So, $R_{\rm S-Restart}$ is a concave function of $r$, when $r{>}\Gamma_{\rm S-Restart}$. Based on {\em Lemma} \ref{lemma3}, we know that $\lg(R_{\rm Clone}-R_{\rm min})$ is a concave function of $r$, when $r{>}\Gamma_{\rm S-Restart}$.
	
	The second order derivative of $-E_{S-Restart}(T)$ is
	\begin{multline}\label{jo7}
	-\int_{D-\tau_{\rm est}}^{t_{\min}}\left(\frac{t_{\min}}{D}\right)^{\beta}\left(\ln\frac{t_{\min}}{\omega}\right)^2{\cdot}\beta^2{\cdot}\left(\frac{t_{\min}}{\omega}\right)^{\beta{\cdot}r}d\omega \\
	-\int_{D-\tau_{\rm est}}^{\infty}\left(\ln\frac{t_{\min}}{\omega}\right)^2{\cdot}\beta^2{\cdot}\left(\frac{t_{\min}}{\omega+\tau_{\rm est}}\right)^{\beta}\left(\frac{t_{\min}}{\omega}\right)^{\beta{\cdot}r}d\omega,
	\end{multline}
	which is less than $0$. Given the summation of two concave functions is also a concave function, thus $U_{\rm S-Restart}(r)$ is a concave function when $r{>}\Gamma_{\rm S-Restart}$.
\end{proof}

\begin{proof}
	In following, we prove that $\lg(R_{\rm S-Resume}(r)-R_{\min})$ is concave when
	\begin{equation}
	r {>} \Gamma_{\rm S-Resume} \beta^{-1}{\cdot}\log_{\frac{\bar{t}_{\min}}{D-\tau_{\rm est}}}\frac{D^{\beta}}{N{\cdot}t^{\beta}_{\min}}-1,
	\end{equation}
	and $-\theta{\cdot}C{\cdot}E_{\rm S-Resume}(T)$ is concave for all values of $r$.
	
	First, we note that $\lg(R_{\rm S-Resume}(r)-R_{\min})$ is an increasing and concave function of $R_{\rm S-Resume}(r)$. Also, when $r{>}\Gamma_{\rm S-Resume}$, the second order derivative of $R_{\rm S-Resume}$ is less than $0$. So, $R_{\rm S-Resume}$ is a concave function of $r$, when $r{>}\Gamma_{\rm Clone}$. Based on {\em Lemma} \ref{lemma3}, we know that $\lg(R_{\rm S-Resume}-R_{\min})$ is a concave function of $r$, when $r{>}\Gamma_{\rm S-Resume}$.
	
	The second order derivative of $-E_{S-Resume}(T)$ equals
	\begin{equation}\label{jo11}
	-\int_{t_{\min}}^{\infty}\left(\frac{t_{\min}}{D}\right)^{\beta}\left(\frac{\bar{t}_{\min}}{t}\right)^{\beta{\cdot}(r+1)} \\
	{\cdot}\left(\ln\frac{\bar{t}_{\min}}{t}\right)^2{\cdot}\beta^2dt
	\end{equation}
	which is less than $0$. Given the summation of two concave functions is also a concave function, thus $U_{\rm S-Resume}(r)$ is a concave function when $r{>}\Gamma_{\rm S-Resume}$.
\end{proof} 
\end{document}